\pdfoutput=1
\documentclass[a4paper,11pt]{article}

\usepackage[hmargin=0.875in,vmargin=1in]{geometry}
\usepackage[utf8]{inputenc}
\usepackage[T1]{fontenc}
\usepackage{lmodern,microtype}
\usepackage{mathtools,amsthm,amssymb,enumitem}
\usepackage{tikz}
\usepackage[unicode]{hyperref}
\usepackage{bookmark}

\hypersetup{
  pdftitle={Approximating pathwidth for graphs of small treewidth},
  pdfauthor={Carla Groenland, Gwenaël Joret, Wojciech Nadara, Bartosz Walczak},
  colorlinks=true,linkcolor=black,citecolor=black,filecolor=black,urlcolor=black
}

\let\origbibliography\thebibliography
\renewcommand\thebibliography[1]{\origbibliography{#1}\itemsep 0pt plus 0.2ex}

\newtheorem{theorem}{Theorem}[section]
\newtheorem{corollary}[theorem]{Corollary}
\newtheorem{lemma}[theorem]{Lemma}
\newtheorem{claim}[theorem]{Claim}
\newtheorem{conjecture}[theorem]{Conjecture}

\setlist[enumerate]{label=\upshape{(\arabic*)},noitemsep,topsep=\smallskipamount,leftmargin=*}
\newlist{enumalph}{enumerate}{1}
\setlist[enumalph]{label=\upshape{(\alph*)},noitemsep,topsep=\smallskipamount,leftmargin=*,widest=d}
\newlist{enumroman}{enumerate}{1}
\setlist[enumroman]{label=(\roman*),noitemsep,topsep=\smallskipamount,leftmargin=*,widest=ii}
\setlist[itemize]{noitemsep,topsep=\smallskipamount,leftmargin=1.5em}

\title{Approximating Pathwidth for Graphs of Small Treewidth\thanks{A preliminary version of this paper appeared in \emph{Proceedings of the Thirty-Second Annual ACM-SIAM Symposium on Discrete Algorithms}, pages 1965--1976. Society for Industrial and Applied Mathematics, Philadelphia, 2021.
That version contains two serious errors---in the proof of Theorem~3.1 and in the complexity analysis of the algorithm.}}

\author{Carla Groenland\thanks{Utrecht University, The Netherlands (\texttt{c.e.groenland@uu.nl}).
Supported by the project CRACKNP (ERC grant agreement No 853234).}
\and
Gwenaël Joret\thanks{Département d'Informatique, Université Libre de Bruxelles, Brussels, Belgium (\texttt{gwenael.joret@ulb.be}).
Research supported by an ARC grant from the Wallonia-Brussels Federation of Belgium, by a CDR grant from the National Fund for Scientific Research (FNRS), and by the Wallonia Brussels International (WBI) agency.}
\and
Wojciech Nadara\thanks{Institute of Informatics, University of Warsaw, Poland (\texttt{w.nadara@mimuw.edu.pl}).
This research is part of a~project that has received funding from the European Research Council (ERC) under the European Union's Horizon 2020 research and innovation programme Grant Agreement 714704.}
\and
Bartosz Walczak\thanks{Department of Theoretical Computer Science, Faculty of Mathematics and Computer Science, Jagiellonian University, Kraków, Poland (\texttt{bartosz.walczak@uj.edu.pl}).
Research partially supported by the Polish National Agency for Academic Exchange (NAWA).}}

\date{}

\newcommand{\setN}{\mathbb{N}}
\newcommand{\calC}{\mathcal{C}}
\newcommand{\calT}{\mathcal{T}}
\newcommand{\CC}[1]{\mathcal{C}_{#1}}

\newcommand{\inputB}{B^{\mathrm{in}}}
\newcommand{\inputG}{G^{\mathrm{in}}}
\newcommand{\inputT}{T^{\mathrm{in}}}
\newcommand{\inputV}{V^{\mathrm{in}}}

\newcommand{\size}[1]{{\lvert #1\rvert}}
\newcommand{\witness}[1]{\langle #1\rangle}
\newcommand{\blowup}[1]{#1^{\smash{(t)}}}

\DeclareMathOperator{\tw}{tw}
\DeclareMathOperator{\pw}{pw}
\DeclareMathOperator{\Sub}{Sub}
\DeclareMathOperator{\solve}{\mathsf{solve}}
\DeclareMathOperator{\level}{level}

\let\leq\leqslant
\let\geq\geqslant

\begin{document}

\maketitle

\begin{abstract}
We describe a polynomial-time algorithm which, given a graph $G$ with treewidth $t$, approximates the pathwidth of $G$ to within a ratio of $O(t\sqrt{\log t})$.
This is the first algorithm to achieve an $f(t)$-approximation for some function $f$.

Our approach builds on the following key insight: every graph with large pathwidth has large treewidth or contains a subdivision of a large complete binary tree.
Specifically, we show that every graph with pathwidth at least $th+2$ has treewidth at least $t$ or contains a subdivision of a complete binary tree of height $h+1$.
The bound $th+2$ is best possible up to a multiplicative constant.
This result was motivated by, and implies (with $c=2$), the following conjecture of Kawarabayashi and Rossman (SODA'18): there exists a universal constant $c$ such that every graph with pathwidth $\Omega(k^c)$ has treewidth at least $k$ or contains a subdivision of a complete binary tree of height $k$.

Our main technical algorithm takes a graph $G$ and some (not necessarily optimal) tree decomposition of $G$ of width $t'$ in the input, and it computes in polynomial time an integer $h$, a certificate that $G$ has pathwidth at least $h$, and a path decomposition of $G$ of width at most $(t'+1)h+1$.
The certificate is closely related to (and implies) the existence of a subdivision of a complete binary tree of height $h$.
The approximation algorithm for pathwidth is then obtained by combining this algorithm with the approximation algorithm of Feige, Hajiaghayi, and Lee (STOC'05) for treewidth.
\end{abstract}

\section{Introduction}

Tree decompositions and path decompositions are fundamental objects in graph theory.
For algorithmic purposes, it would be highly useful to be able to compute such decompositions of minimum width, that is, achieving the treewidth and the pathwidth of the graph, respectively.
However, both these problems are NP-hard, and remain so even when restricted to very specific graph classes \cite{ACP87,Bodlaender92,Gusted93,HM94,KBMK93,KKM95,MS88}.

The best known polynomial-time approximation algorithm for treewidth, due to Feige, Hajiaghayi, and Lee~\cite{FHL08}, computes a tree decomposition of an input graph $G$ whose width is within a factor of $O(\sqrt{\log\tw(G)})$ of the treewidth $\tw(G)$ of $G$.
A modification of that algorithm produces a path decomposition whose width is within a factor of $O(\log n\cdot\sqrt{\log\tw(G)})$ of the pathwidth $\pw(G)$ of $G$, where $n$ is the number of vertices of $G$, using the fact that $\tw(G)\leq\pw(G)\leq(\tw(G)+1)\log_2n$~\cite{FHL08}.
Combining it with existing FPT algorithms for treewidth and pathwidth~\cite{Bodlaender96,BK93} leads to a polynomial-time algorithm approximating pathwidth to within a factor of $O(\pw(G)\tw(G)\sqrt{\log\tw(G)})$.%
\footnote{Here is a brief sketch.
Instances $G$ such that $\pw(G)\tw(G)=O(\log n)$ can be solved optimally in polynomial time thanks to the following two algorithms.
(1) Bodlaender~\cite{Bodlaender96} described an algorithm that, given $G$ and a number $t$, constructs a tree decomposition of $G$ of width at most $t$ or finds that $\tw(G)>t$ in $2^{O(t^2)}n$ time.
(2) Bodlaender and Kloks~\cite{BK93} described an algorithm that, given $G$, a tree decomposition of $G$ of width $t$, and a number $p$, constructs a path decomposition of $G$ of width at most $p$ or finds that $\pw(G)>p$ in $2^{O(pt)}n$ time.
When $\pw(G)\tw(G)=\Omega(\log n)$, the aforementioned $O(\log n\cdot\sqrt{\log\tw(G)})$-approximation algorithm of Feige et~al.\ achieves the claimed bound.}
This is the best known approximation ratio for pathwidth that we are aware of.
A simple polynomial-time $f(\pw(G))$-approximation algorithm (with $f$ exponential) follows from the result of~\cite{CDF96}.

We describe a polynomial-time algorithm that approximates pathwidth to within a factor of $O(\tw(G)\sqrt{\log\tw(G)})$, thus effectively dropping the $\pw(G)$ factor in the above.
To our knowledge, this is the first algorithm to achieve an approximation ratio that depends only on treewidth.

Our approach builds on the following key insight: every graph with large pathwidth has large treewidth or contains a subdivision of a large complete binary tree.

\begin{theorem}
\label{thm1}
Every graph with treewidth\/ $t-1$ has pathwidth at most\/ $th+1$ or contains a subdivision of a complete binary tree of height\/ $h+1$.
\end{theorem}

The bound $th+1$ is best possible up to a multiplicative constant (see Section~\ref{sec:tight}).
Our original motivation for Theorem~\ref{thm1} was the following result of Kawarabayashi and Rossman~\cite{KR22} about treedepth, which is an upper bound on pathwidth: every graph with treedepth $\Omega(k^5\log^2k)$ has treewidth at least $k$, or contains a subdivision of a complete binary tree of height $k$, or contains a path of order $2^k$.
The bound $\Omega(k^5\log^2k)$ was recently lowered to $\Omega(k^3)$ by Czerwiński, Nadara, and Pilipczuk~\cite{CNP21}, who also devised an $O(\tw(G)\log^{3/2}\tw(G))$-approximation algorithm for treedepth.
Kawarabayashi and Rossman~\cite{KR22} conjectured that the third outcome of their theorem, the path of order $2^k$, could be avoided if one considered pathwidth instead of treedepth: they conjectured the existence of a universal constant $c$ such that every graph with pathwidth $\Omega(k^c)$ has treewidth at least $k$ or contains a subdivision of a complete binary tree of height $k$.
Theorem~\ref{thm1} implies their conjecture with $c=2$, which is best possible (see Section~\ref{sec:tight}).
We remark that Wood~\cite{Wood13} also conjectured a statement of this type, with a bound of the form $f(t)\cdot h$ on the pathwidth for some function $f$ (see also \cite[Lemma~6]{MW15} and \cite[Conjecture~6.4.1]{Hickingbotham19}).
Both Theorem~\ref{thm1} and the treedepth results~\cite{KR22,CNP21} are a continuation of a line of research on excluded minor characterizations of graphs with small values of their corresponding width parameters (treewidth/pathwidth/treedepth), which was started by the seminal Grid Minor Theorem~\cite{RS86} and its improved polynomial versions~\cite{CC16,CT21}.
In the last section, we propose yet another problem in a similar vein (see Conjecture~\ref{conjecture}).

Since the complete binary tree of height $h$ has pathwidth $\lceil h/2\rceil$~\cite{Scheffler89}, any subdivision of it (as a subgraph) can be used to certify that the pathwidth of a given graph is large.
The following key concept provides a stronger certificate of large pathwidth, more suitable for our purposes.
Let $(\calT_h)_{h=0}^\infty$ be a sequence of classes of graphs defined inductively as follows: $\calT_0$ is the class of all connected graphs, and $\calT_{h+1}$ is the class of connected graphs $G$ that contain three pairwise disjoint sets of vertices $V_1$, $V_2$, and $V_3$ such that $G[V_1],G[V_2],G[V_3]\in\calT_h$ and any two of $V_1$, $V_2$, and $V_3$ can be connected in $G$ by a path avoiding the third one.
Every graph in $\calT_h$ has the following properties:
\begin{itemize}
\item it has pathwidth at least $h$ (see Lemma~\ref{lem:Th-pw}), and
\item it contains a subdivision of a complete binary tree of height $h$ (see Lemma~\ref{lem:Th-subdivision}).
\end{itemize}

Theorem~\ref{thm1} has a short and simple proof (see Section~\ref{sec:proof1}).
It proceeds by showing that every connected graph with treewidth $t-1$ has pathwidth at most $th+1$ or belongs to $\calT_{h+1}$.
The stronger assertion allows us to apply induction on $h$.
Unfortunately, this proof is not algorithmic.

To obtain the aforementioned approximation algorithm, we prove the following algorithmic version of Theorem~\ref{thm1}.
Its proof is significantly more involved (see Section~\ref{sec:proof2}).

\begin{theorem}
\label{thm2}
For every connected graph\/ $G$ with treewidth at most\/ $t-1$, there is an integer\/ $h\geq 0$ such that\/ $G\in\calT_h$ and\/ $G$ has pathwidth at most\/ $th+1$.
Moreover, there is a polynomial-time algorithm to compute such an integer\/ $h$, a path decomposition of\/ $G$ of width at most\/ $th+1$, and a subdivision of a complete binary tree of height\/ $h$ in\/ $G$ given a tree decomposition of\/ $G$ of width at most\/ $t-1$.
\end{theorem}

Since every graph in $\calT_h$ has pathwidth at least $h$, combining Theorem~\ref{thm2} (applied to every connected component of the input graph) with the aforementioned approximation algorithm for treewidth of Feige et al.~\cite{FHL08}, we obtain the following approximation algorithm for pathwidth.

\begin{corollary}
\label{cor:approx}
There is a polynomial-time algorithm which, given a graph\/ $G$ of treewidth\/ $t$ and pathwidth\/ $p$, computes a path decomposition of\/ $G$ of width\/ $O(t\sqrt{\log t}\cdot p)$.
Moreover, if a tree decomposition of\/ $G$ of width\/ $t'$ is also given in the input, the resulting path decomposition has width at most\/ $(t'+1)p+1$.
\end{corollary}

We remark that if we consider graphs $G$ coming from a fixed class of graphs with bounded treewidth, then we can first use an algorithm of Bodlaender~\cite{Bodlaender96} to compute an optimal tree decomposition of $G$ in linear time, and then use the above algorithm to approximate pathwidth to within a ratio of roughly $\tw(G)+1$.
We note the following two precursors of this result in the literature (with slightly better approximation ratios): Bodlaender and Fomin~\cite{BF02} gave a $2$-approximation algorithm for computing the pathwidth of outerplanar graphs (a subclass of graphs of treewidth at most $2$), and Fomin and Thilikos~\cite{FT06} gave a $3$-approximation algorithm for computing the pathwidth on Halin graphs (a subclass of graphs of treewidth at most $3$).

We conclude this introduction with a remark about parameterized algorithms, even though our focus in this paper is approximation algorithms with running time polynomial in the size of the input graph.
Bodlaender~\cite{Bodlaender96} (see also~\cite{Bodlaender12}) designed a linear-time FPT algorithm for computing pathwidth when parameterized by the pathwidth.
That is, for an $n$-vertex input graph $G$, his algorithm computes the pathwidth $\pw(G)$ and an optimal path decomposition of $G$ in $f(\pw(G))\cdot n$ time for some computable function $f$.
Bodlaender and Kloks~\cite{BK93} considered the problem of computing the pathwidth when the input graph has small treewidth.
They devised an XP algorithm for computing pathwidth when parameterized by the treewidth: given an $n$-vertex graph $G$, the algorithm computes $\pw(G)$ and an optimal path decomposition of $G$ in $n^{f(\tw(G))}$ time for some computable function $f$.
It is an old open problem whether pathwidth is fixed-parameter tractable when parameterized by the treewidth, that is, whether there exists an algorithm to compute the pathwidth of an $n$-vertex input graph $G$ in $f(\tw(G))\cdot n^{O(1)}$ time for some computable function $f$.
This question was first raised by Dean~\cite{Dean93}.
Fomin and Thilikos~\cite{FT06} pointed out that even obtaining an approximation of pathwidth when parameterized by treewidth is open.
Our approximation algorithm is a solution to the latter problem (in a strong sense---with polynomial dependence of the running time in the parameter) and can be seen as a step in the direction of Dean's question.

\section{Preliminaries and Tools}

\subsection{Basic Definitions}

Graphs considered in this paper are finite, simple, and undirected.
We use standard graph-theoretic terminology and notation.
We allow a graph to have no vertices; by convention, such a graph \emph{is not connected} and has no connected components.
The vertex set of a graph $G$ is denoted by $V(G)$.
A vertex $v$ of a graph $G$ is considered a neighbor of a set $X\subseteq V(G)$ if $v\notin X$ and $v$ is connected by an edge to some vertex in $X$.
The neighborhood (thus defined) of a set $X$ in $G$ is denoted by $N_G(X)$.

A \emph{tree decomposition} of a graph $G$ is a pair $(T,\{B_x\}_{x\in V(T)})$ where $T$ is a tree and $\{B_x\}_{x\in V(T)}$ is a family of subsets of $V(G)$ called \emph{bags}, satisfying the following two conditions:
\begin{itemize}
\item for each vertex $v$ of $G$, the set of nodes $\{x\in V(T)\colon v\in B_x\}$ induces a non-empty subtree of $T$;
\item for each edge $uv$ of $G$, there is a node $x$ in $T$ such that both $u$ and $v$ belong to $B_x$.
\end{itemize}
The \emph{width} of a tree decomposition $(T,\{B_x\}_{x\in V(T)})$ is $\max_{x\in V(T)}\size{B_x}-1$.
The \emph{treewidth} of a graph $G$ is the minimum width of a tree decomposition of $G$.
A \emph{path decomposition} and \emph{pathwidth} are defined analogously with the extra requirement that the tree $T$ is a path.
The treewidth and the pathwidth of a graph $G$ are denoted by $\tw(G)$ and $\pw(G)$, respectively.
We refer the reader to~\cite{Diestel10} for background on tree decompositions.

A \emph{complete binary tree of height} $h$ is a rooted tree in which every non-leaf node has two children and every path from the root to a leaf has $h$ edges.
Such a tree has $2^{h+1}-1$ nodes.
A \emph{complete ternary tree of height} $h$ is defined analogously but with the requirement that every non-leaf node has three children.
A \emph{subdivision} of a tree $T$ is a tree obtained from $T$ by replacing each edge $uv$ with some path connecting $u$ and $v$ whose internal nodes are new nodes private to that path.

\subsection{Witnesses for Large Pathwidth}

Recall that $(\calT_h)_{h=0}^\infty$ is the sequence of classes of graphs defined inductively as follows: $\calT_0$ is the class of all connected graphs, and $\calT_{h+1}$ is the class of connected graphs $G$ that contain three pairwise disjoint sets of vertices $V_1$, $V_2$, and $V_3$ such that $G[V_1],G[V_2],G[V_3]\in\calT_h$ and any two of $V_1$, $V_2$, and $V_3$ can be connected in $G$ by a path avoiding the third one.

A \emph{$\calT_h$-witness} for a graph $G\in\calT_h$ is a complete ternary tree of height $h$ of subsets of $V(G)$ defined inductively following the definition of $\calT_h$.
The $\calT_0$-witness for a connected graph $G$ is the tree with the single node $V(G)$, denoted by $\witness{V(G)}$.
A $\calT_{h+1}$-witness for a graph $G\in\calT_{h+1}$ is a tree with root $V(G)$ and with three subtrees $W_1,W_2,W_3$ of the root that are $\calT_h$-witnesses of $G[V_1],G[V_2],G[V_3]$ for some sets $V_1,V_2,V_3$ as in the definition of $\calT_{h+1}$; it is denoted by $\witness{V(G);W_1,W_2,W_3}$.

It clearly follows from these definitions that every graph in $\calT_h$ has at least $3^h$ vertices and every $\calT_h$-witness of an $n$-vertex graph has $O(n)$ nodes.
The next two lemmas explain the connection of $\calT_h$ to pathwidth and to subdivisions of complete binary trees.

\begin{lemma}
\label{lem:Th-pw}
If\/ $G\in\calT_h$, then\/ $\pw(G)\geq h$.
\end{lemma}

\begin{proof}
The proof goes by induction on $h$.
The case $h=0$ is trivial.
Now, assume that $h\geq 1$ and the lemma holds for $h-1$.
Since $G\in\calT_h$, there are sets $V_1,V_2,V_3\subseteq V(G)$ interconnected as in the definition of $\calT_h$, such that $G[V_i]\in\calT_{h-1}$ and thus $\pw(G[V_i])\geq h-1$ for $i=1,2,3$.
Let $P$ be a path decomposition of $G$.
With bags restricted to $V_i$, it becomes a path decomposition of $G[V_i]$.
It follows that for $i=1,2,3$, there is a bag $B_i$ in $P$ such that $\size{B_i\cap V_i}\geq h$.
Assume without loss of generality that $B_1,B_2,B_3$ occur in this order in $P$.
Since $G[V_1]$ and $G[V_3]$ are connected, there is a path that connects $B_1\cap V_1$ and $B_3\cap V_3$ in $G$ avoiding $V_2$.
This path must have a vertex in $B_2$, so $\size{B_2\setminus V_2}\geq 1$ and thus $\size{B_2}\geq h+1$.
This proves that $\pw(G)\geq h$.
\end{proof}

The proof of Lemma~\ref{lem:Th-pw} generalizes the well-known proof of the fact that (a subdivision of) a complete binary tree of height $h$ has pathwidth at least $\lceil h/2\rceil$.
Actually, it is straightforward to show that such a tree belongs to $\calT_{\lceil h/2\rceil}$.

\begin{lemma}
\label{lem:Th-subdivision}
If\/ $G\in\calT_h$, then\/ $G$ contains a subdivision of a complete binary tree of height\/ $h$ as a subgraph.
Moreover, it can be computed in polynomial time from a\/ $\calT_h$-witness for\/ $G$.
\end{lemma}

\begin{proof}
We prove, by induction on $h$, that for every graph $G\in\calT_h$ and every $v\in V(G)$, the following structure exists in $G$: a subdivision $S$ of a complete binary tree of height $h$ with some root $r$ and a path $P$ from $v$ to $r$ such that $V(P)\cap V(S)=\{r\}$.
This is trivial for $h=0$.
For the induction step, assume that $h\geq 1$ and the statement holds for $h-1$.
Let $G\in\calT_h$ and $v\in V(G)$.
Let $V_1,V_2,V_3\subseteq V(G)$ be as in the definition of $\calT_h$.
Assume without loss of generality that $v\in V_3$ or $v$ can be connected with $V_3$ by a path in $G$ avoiding $V_1\cup V_2$.
For $i=1,2$, since $G[V_3]$ is connected and $G$ has a path connecting $V_i$ with $V_3$ and avoiding $V_{3-i}$, there is also a path in $G$ from $v$ to some vertex $v_i\in V_i$ avoiding $V_1\cup V_2\setminus\{v_i\}$.
These paths can be chosen so that they first follow a common path $P$ from $v$ to some vertex $r$ in $G-(V_1\cup V_2)$ and then they split into a path $Q_1$ from $r$ to $v_1$ and a path $Q_2$ from $r$ to $v_2$ so that $r$ is the only common vertex of any two of $P,Q_1,Q_2$.
For $i=1,2$, the induction hypothesis provides an appropriate structure in $G[V_i]$: a subdivision $S_i$ of a complete binary tree of height $h-1$ with root $r_i$ and a path $P_i$ from $v_i$ to $r_i$ such that $V(P_i)\cap V(S_i)=\{r_i\}$.
Connecting $r$ with $S_1$ and $S_2$ by the combined paths $Q_1P_1$ and $Q_2P_2$, respectively, yields a subdivision $S$ of a complete binary tree of height $h$ with root $r$ in $G$.
The construction guarantees that $V(P)\cap V(S)=\{r\}$.

Clearly, given a $\calT_h$-witness for $G$, the induction step described above can be performed in polynomial time, and therefore the full recursive procedure of computing a subdivision of a complete binary tree of height $h$ in $G$ works in polynomial time.
\end{proof}

\subsection{Combining Path Decompositions}

The following lemma will be used several times in the paper to combine path decompositions.

\begin{lemma}
\label{lem:combine}
Let\/ $G$ be a graph and\/ $(T,\{B_x\}_{x\in V(T)})$ be a tree decomposition of\/ $G$ of width\/ $t-1$.
\begin{enumerate}
\item\label{item:combine-vertex} If\/ $q\in V(T)$ and every connected component of\/ $G-B_q$ has pathwidth at most\/ $\ell$, then there is a path decomposition of\/ $G$ of width at most\/ $\ell+t$ which contains\/ $B_q$ in every bag.
\item\label{item:combine-subpath} If\/ $Q$ is the path connecting\/ $x$ and\/ $y$ in\/ $T$ and every connected component of\/ $G-\bigcup_{q\in V(Q)}B_q$ has pathwidth at most\/ $\ell$, then there is a path decomposition of\/ $G$ of width at most\/ $\ell+t$ which contains\/ $B_x$ in the first bag and\/ $B_y$ in the last bag.
\end{enumerate}
In either case, there is a polynomial-time algorithm to construct such a path decomposition of\/ $G$ from the path decompositions of the respective components\/ $C$ of width at most\/ $\ell$.
\end{lemma}

\begin{proof}
In case~\ref{item:combine-vertex}, the path decomposition of $G$ is obtained by concatenating the path decompositions of the connected components of $G-B_q$ (which have width at most $\ell$) and adding $B_q$ to every bag.
Now, consider case~\ref{item:combine-subpath}.
For every node $q$ of $Q$, let $T_q$ be the subtree of $T$ induced on the nodes $z$ such that the path from $q$ to $z$ in $T$ contains no other nodes of $Q$, and let $V_q=\bigcup_{z\in V(T_q)}B_z$.
Apply case~\ref{item:combine-vertex} to the graph $G[V_q]$, the tree decomposition $(T_q,\{B_z\}_{z\in V(T_q)})$ of $G[V_q]$, and the node $q\in V(T_q)$ to obtain a path decomposition of $G[V_q]$ of width at most $\ell+t$ containing $B_q$ in every bag.
Then, concatenate the path decompositions thus obtained for all nodes $q$ of $Q$ (in the order they occur on $Q$) to obtain a requested path decomposition of $G$.
\end{proof}

\section{Proof of Theorem \ref{thm1}}
\label{sec:proof1}

The statement of Theorem~\ref{thm1} on a graph $G$ follows from the same statement on every connected component of $G$.
By Lemma~\ref{lem:Th-subdivision}, every graph in $\calT_h$ contains a subdivision of a complete binary tree of height $h$.
Thus, Theorem~\ref{thm1} is a direct corollary to the following statement.

\begin{theorem}
\label{thm:existence}
For every\/ $h\in\setN$, every connected graph with treewidth at most\/ $t-1$ has pathwidth at most\/ $th+1$ or belongs to\/ $\calT_{h+1}$.
\end{theorem}

A tree decomposition of $G$ is \emph{optimal} if its width is equal to $\tw(G)$.
For the proof of Theorem~\ref{thm:existence}, we need optimal tree decompositions with an additional property.
Namely, consider a connected graph $G$ and a tree decomposition $(T,\{B_x\}_{x\in V(T)})$ of $G$.
Every edge $xy$ of $T$ splits $T$ into two subtrees: $T_{x|y}$ containing $x$ and $T_{y|x}$ containing $y$.
For every oriented edge $xy$ of $T$, let $G_{x|y}$ denote the subgraph of $G$ induced on the union of the bags of the nodes in $T_{x|y}$.
The property we need is that every subgraph of the form $G_{x|y}$ is connected.
It is known that such a tree decomposition always exists~\cite{FN06}, but for completeness, we present a short proof of this fact in the following lemma.

\begin{lemma}
\label{lem:connected-decomposition}
Every connected graph\/ $G$ has an optimal tree decomposition\/ $(T,\{B_x\}_{x\in V(T)})$ with the property that\/ $G_{x|y}$ is connected for every oriented edge\/ $xy$ of\/ $T$.
\end{lemma}

\begin{proof}
Let $t=\tw(G)+1$.
The \emph{fatness} of an optimal tree decomposition $(T,\{B_x\}_{x\in V(T)})$ of $G$ is the $t$-tuple $(k_0,\ldots,k_{t-1})$ such that $k_i$ is the number of bags $B_x$ of size $t-i$.
Let $(T,\{B_x\}_{x\in V(T)})$ be an optimal tree decomposition of $G$ with lexicographically minimum fatness.
(The idea of taking such a tree decomposition comes from the proof of existence of optimal ``lean'' tree decompositions due to Bellenbaum and Diestel \cite[Theorem~3.1]{BD02}.)
We show that it has the required property.

Suppose it does not.
Let $xy$ be an edge of $T$ such that $G_{x|y}$ is disconnected and the number of nodes in the subtree $T_{x|y}$ of $T$ is minimized.
Let $\calC$ be the family of connected components of $G_{x|y}$, so that $\size{\calC}\geq 2$.
Let $Z=N_T(x)\setminus\{y\}$.
For every node $z\in Z$, let $C_z$ be the connected component of $G_{x|y}$ that contains $G_{z|x}$, which exists because the choice of $xy$ guarantees that $G_{z|x}$ is connected.

We modify $(T,\{B_x\}_{x\in V(T)})$ into a new tree decomposition of $G$ as follows.
We keep all nodes other than $x$ (with their bags $B_x$) and all edges non-incident to $x$.
We replace the node $x$ by $\size{\calC}$ nodes $x_C$ with bags $B_{x_C}=B_x\cap V(C)$ for each $C\in\calC$.
We replace the edge $xy$ by $\size{\calC}$ edges $x_Cy$ for each $C\in\calC$, and we replace the edge $xz$ by an edge $x_{C_z}z$ for each $z\in Z$.
It is straightforward to verify that what we obtain is indeed a tree decomposition of $G$ and it has width $t$.

Since $G$ is connected, we have $B_{x_C}=B_x\cap V(C)\neq\emptyset$ for every $C\in\calC$.
This and the assumption that $\size{\calC}\geq 2$ imply that $\size{B_{x_C}}<\size{B_x}$ for all $C\in\calC$.
We conclude that the fatness of the new tree decomposition is lexicographically less than the fatness of $(T,\{B_x\}_{x\in V(T)})$, which contradicts the assumption that the latter is lexicographically minimal.
\end{proof}

\begin{proof}[Proof of Theorem~\ref{thm:existence}]
The proof goes by induction on $h$.
The statement is true for $h=0$: if a connected graph $G$ has a cycle or a vertex of degree at least $3$, then $G\in\calT_1$, and otherwise $G$ is a path, so $\pw(G)\leq 1$.
For the rest of the proof, assume that $h\geq 1$ and the statement is true for $h-1$.

Let $G$ be a connected graph width treewidth at most $t-1$ and $(T,\{B_x\}_{x\in V(T)})$ be an optimal tree decomposition of $G$ obtained from Lemma~\ref{lem:connected-decomposition}.
Thus $\size{B_x}\leq t$ for every node $x$ of $T$ and $G_{x|y}$ is connected for every oriented edge $xy$ of $T$.
For every oriented edge $xy$ of $T$, let $F_{x|y}$ be the subgraph of $G$ induced on the vertices not in $G_{y|x}$, that is, on the vertices that belong only to bags in $T_{x|y}$ and to no other bags.
Let $E$ be the set of edges $xy$ of $T$ such that both $F_{x|y}$ and $F_{y|x}$ have a connected component belonging to $\calT_h$.

Suppose $E=\emptyset$.
It follows that every pair of trees of the form $T_{x|y}$ such that $F_{x|y}$ has a connected component in $\calT_h$ has a common node.
This implies that all such trees have a common node, say $z$, by the well-known fact that subtrees of a tree have the Helly property \cite[Theorem~4.1]{Horn72}.
For every neighbor $y$ of $z$ in $T$ and every connected component $C$ of $F_{y|z}$, since $C\notin\calT_h$, the induction hypothesis gives $\pw(C)\leq t(h-1)+1$.
Lemma~\ref{lem:combine}~\ref{item:combine-vertex} applied with $q=z$ yields $\pw(G)\leq th+1$.

For the rest of the proof, assume $E\neq\emptyset$.
Since every connected supergraph of a graph from $\calT_h$ belongs to $\calT_h$, the set $E$ is the edge set of some subtree $K$ of $T$.
Let $Z$ be the set of leaves of $K$.
Since $K$ has at least one edge, we have $\size{Z}\geq 2$.

Suppose $\size{Z}\geq 3$.
Choose any distinct $z_1,z_2,z_3\in Z$.
For each $i\in\{1,2,3\}$, let $C_i$ be a connected component of $F_{z_i|x_i}$ that belongs to $\calT_h$, where $x_i$ is the unique neighbor of $z_i$ in $K$.
For each $i\in\{1,2,3\}$, the subgraph $G_{x_i|z_i}$ is connected, is vertex-disjoint from $C_i$, and contains the other two of $C_1$, $C_2$, and $C_3$.
Consequently, any two of the sets $V(C_1)$, $V(C_2)$, and $V(C_3)$ can be connected by a path in $G$ avoiding the third one.
This shows that $G\in\calT_{h+1}$.

Now, suppose $\size{Z}=2$.
It follows that $K$ is a path $x_1\ldots x_m$, where $Z=\{x_1,x_m\}$.
For every node $x_i$ of $K$, every neighbor $y$ of $x_i$ in $T-K$, and every connected component $C$ of $F_{y|x_i}$, since $C\notin\calT_h$, the induction hypothesis gives $\pw(C)\leq t(h-1)+1$.
Lemma~\ref{lem:combine}~\ref{item:combine-subpath} applied with $Q=K$ yields $\pw(G)\leq th+1$.
\end{proof}

\section{Proof of Theorem \ref{thm2}}
\label{sec:proof2}

By Lemma~\ref{lem:Th-subdivision}, every graph in $\calT_h$ contains a subdivision of a complete binary tree of height $h$, which can be computed in polynomial time from a $\calT_h$-witness of $G$.
Therefore, Theorem~\ref{thm2} is a direct corollary to the following statement, the proof of which is presented in this section.

\begin{theorem}
\label{thm:algorithm}
There is a polynomial-time algorithm which, given a connected graph\/ $G$ and a tree decomposition of\/ $G$ of width at most\/ $t-1$, computes
\begin{itemize}
\item a number\/ $h\in\setN$ such that\/ $G\in\calT_h$ and\/ $\pw(G)\leq th+1$,
\item a\/ $\calT_h$-witness for\/ $G$,
\item a path decomposition of\/ $G$ of width at most\/ $th+1$.
\end{itemize}
\end{theorem}

\subsection{Normalized Tree Decompositions}
\label{subsec:normalized}

Let $G$ be a graph with a fixed rooted tree decomposition $(T',\{B'_x\}_{x\in V(T')})$ of width at most $t-1$, which we call the \emph{initial tree decomposition} of $G$.
For a node $x$ of $T'$, let $T'_x$ be the subtree of $T'$ consisting of $x$ and all nodes of $T'$ lying below $x$, and let $V'_x$ be the set of vertices of $G$ contained only in bags of $T'_x$ (i.e., in no bags of $T'-T'_x$).
We show how to turn $(T',\{B'_x\}_{x\in V(T')})$ into a \emph{normalized tree decomposition} of $G$.
Intuitively, the latter will be a cleaned-up version of the initial tree decomposition with some additional properties that will be useful in the algorithm.

For a subset $X$ of the vertex set of a graph $G$, let $\CC{G}(X)$ denote the family of connected components of the induced subgraph $G[X]$ (which is empty when $X=\emptyset$).
Let
\[\Sub(G)=\bigcup_{x\in V(T')}\CC{G}(V'_x).\]
For a connected graph $G$, the set $\Sub(G)$ will be the node set of the normalized tree decomposition of $G$.
We will use Greek letters $\alpha$, $\beta$, etc.\ to denote members of $\Sub(G)$ (nodes of the normalized tree decomposition of $G$).
Here are some easy consequences of these definitions and the assumption that $(T',\{B'_x\}_{x\in V(T')})$ is a tree decomposition of $G$.

\begin{lemma}
\label{lem:Sub}
The following holds for every graph\/ $G$ and for\/ $\Sub(G)$ defined as above.
\begin{enumerate}
\item If\/ $G$ is connected, then\/ $G\in\Sub(G)$.
\item\label{item:Sub-laminar} If\/ $\alpha,\beta\in\Sub(G)$, then\/ $V(\alpha)\subseteq V(\beta)$, or\/ $V(\beta)\subseteq V(\alpha)$, or\/ $V(\alpha)\cap V(\beta)=\emptyset$.
\item If\/ $\alpha,\beta\in\Sub(G)$ and\/ $V(\alpha)\cap V(\beta)=\emptyset$, then no edge of\/ $G$ connects\/ $V(\alpha)$ and\/ $V(\beta)$.
\end{enumerate}
\end{lemma}

Now, assume that $G$ is a connected graph.
By Lemma~\ref{lem:Sub}, the members of $\Sub(G)$ are organized in a rooted tree with root $G$ and with the following properties for any nodes $\alpha,\beta\in\Sub(G)$:
\begin{itemize}
\item if $\beta$ is a descendant of $\alpha$ (i.e., $\alpha$ lies on the path from the root to $\beta$), then $V(\beta)\subseteq V(\alpha)$;
\item if neither of $\alpha$ and $\beta$ is a descendant of the other, then $V(\alpha)$ and $V(\beta)$ are disjoint and non-adjacent in $G$.
\end{itemize}
Let $T$ denote this rooted tree (not to be confused with $T'$).
The normalized tree decomposition will be built on this tree $T$.
For each $\alpha\in\Sub(G)$, let $A_\alpha$ be the set of vertices $v\in V(G)$ for which $\alpha$ is the smallest subgraph in $\Sub(G)$ containing $v$ (which must be unique), and let $B_\alpha=A_\alpha\cup N_G(V(\alpha))$.
(Recall that $N_G(V(\alpha))\subseteq V(G)\setminus V(\alpha)$ by the definition of neighborhood.)

\begin{lemma}
\label{lem:normalized}
The following holds for every connected graph\/ $G$.
\begin{enumerate}
\item\label{item:normalized-A} $\{A_\alpha\}_{\alpha\in\Sub(G)}$ is a partition of\/ $V(G)$ into non-empty sets.
\item\label{item:normalized-td} $(T,\{B_\alpha\}_{\alpha\in\Sub(G)})$ is a tree decomposition of\/ $G$.
\item\label{item:normalized-width} $\size{B_\alpha}\leq t$ for every\/ $\alpha\in\Sub(G)$.
\end{enumerate}
\end{lemma}

\begin{proof}
It is clear from the definition that $\bigcup_{\alpha\in\Sub(G)}A_\alpha=V(G)$.
Let $\alpha\in\Sub(G)$.
Since $\alpha$ is connected and the vertex sets of the children of $\alpha$ in $T$ are pairwise disjoint and non-adjacent, at least one vertex of $\alpha$ is not a vertex of any child of $\alpha$ and thus belongs to $A_\alpha$.
This shows \ref{item:normalized-A}.

For the proof of \ref{item:normalized-td}, let $\alpha\in\Sub(G)$ and $v\in A_\alpha$.
Then $v\in B_\alpha$.
If $v\in B_\beta$ for some other $\beta\in\Sub(G)$, then $v$ must be a neighbor of $V(\beta)$ in $G$, which implies that $\beta$ is a descendant of $\alpha$ in $T$.
In that case, $v$ is also a neighbor of $V(\gamma)$ (and thus $v\in B_\gamma$) for every internal node $\gamma$ on the path from $\alpha$ to $\beta$ in $T$.
This shows that the nodes $\beta$ of $T$ such that $v\in B_\beta$ form a non-empty connected subtree of $T$.
It remains to show that any two adjacent vertices $u$ and $v$ of $G$ belong to a common bag $B_\alpha$.
Let $\alpha$ be a minimal member of $\Sub(G)$ containing at least one of $u$ and $v$; say, it contains $v$.
Then $v\in A_\alpha$.
If $u\notin A_\alpha$, then $u\notin V(\alpha)$, by minimality of $\alpha$, so $u$ is a neighbor of $V(\alpha)$ in $G$.
In either case, $u,v\in B_\alpha$.

For the proof of \ref{item:normalized-width}, let $\alpha\in\Sub(G)$, and let $x$ be the lowest node of $T'$ such that $\alpha\in\CC{G}(V'_x)$.
Every vertex $v\in A_\alpha$ belongs to $B'_x$; otherwise it would belong to $V'_y$ for some child $y$ of $x$ in $T'$, and the connected component of $G[V'_y]$ containing $v$ would be a proper induced subgraph of $\alpha$, contradicting $v\in A_\alpha$.
Now, let $v$ be a neighbor of $V(\alpha)$ in $G$.
Thus $v$ is a neighbor of some vertex $u\in V(\alpha)$ while $v\notin V(\alpha)$.
It follows that $u\in V'_x$ while $v\notin V'_x$, which implies that $v$ belongs to some bag of $T'-T'_x$ as well as some bag of $T'_x$ (as $uv$ is an edge of $G$), so it belongs to $B'_x$.
This shows that $B_\alpha\subseteq B'_x$, so $\size{B_\alpha}\leq\size{B'_x}\leq t$.
\end{proof}

We call $(T,\{B_\alpha\}_{\alpha\in\Sub(G)})$ the \emph{normalized tree decomposition} of $G$.
By Lemma~\ref{lem:normalized}, it has at most $\size{V(G)}$ nodes and its width is at most $t-1$.
The following lemma is a direct consequence of the construction of the normalized tree decomposition and will be used in the next subsection.
We emphasize that by a \emph{subtree} of the rooted tree $T$ we simply mean a connected subgraph of $T$, that is, a subtree does not need to be closed under taking descendants.

\begin{lemma}
\label{lem:component}
If\/ $Q$ is a subtree of\/ $T$ and\/ $\gamma$ is a connected component of\/ $G-\bigcup_{\xi\in V(Q)}B_\xi$, then either
\begin{enumalph}
\item $\gamma$ is a node of\/ $T-Q$ that is a child in\/ $T$ of some node of\/ $Q$, or
\item $V(\gamma)\cap V(\xi)=\emptyset$ where\/ $\xi$ is the root of\/ $Q$ in\/ $T$, that is, the node of\/ $Q$ closest to the root in\/ $T$.
\end{enumalph}
\end{lemma}

The normalized tree decomposition depends on the choice of the initial tree decomposition for $G$.
In the algorithm, the graphs $G$ considered are induced subgraphs of a common \emph{input graph} $\inputG$ given with an \emph{input tree decomposition} $(\inputT,\{\inputB_x\}_{x\in V(\inputT)})$, and the initial tree decomposition $(T',\{B'_x\}_{x\in V(T')})$ of $G$ considered above in the definition of $\Sub(G)$ is the input tree decomposition of $\inputG$ with appropriately restricted bags (some of which may become empty):
\[(T',\{B'_x\}_{x\in V(T')})=(\inputT,\{\inputB_x\cap V(G)\}_{x\in V(\inputT)}).\]
The fact that the normalized tree decompositions for all induced subgraphs $G$ of $\inputG$ considered in the algorithm come from a common input tree decomposition of $\inputG$ has the following consequences, which will be used in the complexity analysis of the algorithm in Subsection~\ref{subsec:complexity}.
(We emphasize again that, in the following lemma and later in the paper, $\Sub(G)$ is always computed with respect to the initial tree decomposition that is the restriction of the input tree decomposition to $G$.)

\begin{lemma}
\label{lem:Tin}
The following holds for every induced subgraph\/ $G$ of the input graph\/ $\inputG$.
\begin{enumerate}
\item\label{item:Tin-restrict} If\/ $\alpha\in\Sub(G)$ and\/ $V(\alpha)\subseteq X\subseteq V(G)$, then\/ $\alpha\in\Sub(G[X])$.
\item\label{item:Tin-laminar} If\/ $\alpha,\beta\in\Sub(G)$, then\/ $\alpha\in\Sub(\beta)$, or\/ $\beta\in\Sub(\alpha)$, or\/ $V(\alpha)\cap V(\beta)=\emptyset$.
\item\label{item:Tin-recursion} If\/ $\alpha\in\Sub(G)$, then\/ $\Sub(\alpha)=\{\beta\in\Sub(G)\colon V(\beta)\subseteq V(\alpha)\}$.
\item\label{item:Tin-component} If\/ $\alpha\in\Sub(G)$, $X\subseteq V(G)$, and\/ $\gamma\in\CC{G}(V(\alpha)\cap X)$, then\/ $\gamma\in\Sub(G[X])$.
\end{enumerate}
\end{lemma}

\begin{proof}
Let $\inputT_x$ and $\inputV_x$ be defined like $T'_x$ and $V'_x$ but for the tree decomposition $(\inputT,\{\inputB_x\}_{x\in V(\inputT)})$ of $\inputG$.
Let $G$ be an induced subgraph of $\inputG$.
The fact that $\Sub(G)$ is computed with respect to the initial tree decomposition that is the restriction of the input tree decomposition to $G$ implies
\[\Sub(G)=\bigcup_{x\in V(\inputT)}\CC{\inputG}(\inputV_x\cap V(G)).\]

If $\alpha\in\Sub(G)$ and $V(\alpha)\subseteq X\subseteq V(G)$, then the fact that $\alpha$ is a connected component of $\inputG[\inputV_x\cap V(G)]$ for some node $x$ of $\inputT$ implies that it is also a connected component of $\inputG[\inputV_x\cap X]$, which in turn implies $\alpha\in\Sub(G[X])$.
This shows \ref{item:Tin-restrict}.

If $\alpha,\beta\in\Sub(G)$ and $V(\alpha)\subseteq V(\beta)$, then property~\ref{item:Tin-restrict} with $X=V(\beta)$ yields $\alpha\in\Sub(\beta)$.
This implies \ref{item:Tin-laminar} by Lemma~\ref{lem:Sub}~\ref{item:Sub-laminar}, and it also implies the ``$\supseteq$'' inclusion in \ref{item:Tin-recursion} (with $\alpha$ and $\beta$ interchanged).
For the converse inclusion in \ref{item:Tin-recursion}, let $\alpha\in\Sub(G)$ and $\beta\in\Sub(\alpha)$.
Thus $V(\beta)\subseteq V(\alpha)$.
Let $x$ be a node of $\inputT$ such that $\alpha$ is a connected component of $\inputG[\inputV_x\cap V(G)]$.
Since $V(\alpha)\subseteq\inputV_x$, there is a node $y$ in $\inputT_x$ (implying $\inputV_y\subseteq\inputV_x$) such that $\beta$ is a connected component of $\inputG[\inputV_y\cap V(\alpha)]$.
It follows that $\beta$ is a connected component of $\inputG[\inputV_y\cap V(G)]$, so $\beta\in\Sub(G)$.

Finally, if $\alpha\in\Sub(G)$, $X\subseteq V(G)$, and $\gamma\in\CC{G}(V(\alpha)\cap X)$, then $\alpha$ is a connected component of $\inputG[\inputV_x\cap V(G)]$ for some node $x$ of $\inputT$, whence it follows that $\gamma$ is a connected component of $\inputG[\inputV_x\cap X]$ and thus $\gamma\in\Sub(G[X])$.
This shows \ref{item:Tin-component}.
\end{proof}

\subsection{Main Procedure}

The core of the algorithm is a recursive procedure $\solve(G,b)$, where $G$ is a connected graph with $\tw(G)\leq t-1$ and $b\in\setN\cup\{\infty\}$ is an \emph{upper bound request}.
It computes the following data:
\begin{itemize}
\item a number $h=h(G,b)\in\setN$ such that $h\leq b$,
\item a $\calT_h$-witness $W(G,b)$ of $G$,
\item only when $h<b$: a path decomposition $P(G,b)$ of $G$ of width at most $th+1$.
\end{itemize}
The algorithm uses memoization to compute these data only once for each pair $(G,b)$.
A run of $\solve(G,\infty)$ produces the outcome requested in Theorem~\ref{thm:algorithm}.

The purpose of the upper bound request is optimization---we tell the procedure that if it can provide a $\calT_b$-witness for $G$, then we no longer need any path decomposition for $G$.
This allows the procedure to save some computation, perhaps preventing many unnecessary recursive calls.
Our complexity analysis of the algorithm will rely on this optimization.

Below, we present the procedure $\solve(G,b)$ for a fixed connected graph $G$ and a fixed upper bound request $b\in\setN\cup\{\infty\}$.
The procedure assumes access to the normalized tree decomposition $(T,\{B_\alpha\}_{\alpha\in\Sub(G)})$ of $G$ of width at most $t-1$ obtained from some initial tree decomposition of $G$ as described in the definition of $\Sub(G)$.
In the next subsection, we show that a full run of $\solve(\inputG,\infty)$ on an input graph $\inputG$ makes recursive calls to $\solve(G,b)$ for only polynomially many distinct pairs $(G,b)$ if the normalized tree decompositions of all induced subgraphs $G$ of $\inputG$ occurring in these calls are obtained from a common input tree decomposition of $\inputG$ as described in Subsection~\ref{subsec:normalized}.

If $b=0$, then we just set $h(G,0)=0$ and $W(G,0)=\witness{V(G)}$, and we terminate the run of $\solve(G,b)$, saying that it is \emph{pruned}.
Assume henceforth that $b\geq 1$.

Suppose $T$ has only one node, that is, $\Sub(G)=\{G\}$.
Since $V(G)$ is the bag of that node, $\size{V(G)}\leq t$.
If $G$ has a cycle or a vertex of degree at least $3$, then it has three vertices $v_1,v_2,v_3$ such that any two of them can be connected by a path in $G$ avoiding the third one.
In that case, we set $h(G,b)=1$ and $W(G,b)=\witness{V(G);\witness{\{v_1\}},\witness{\{v_2\}},\witness{\{v_3\}}}$, and (if $b>1$) we let $P(G,b)$ be the path decomposition of $G$ consisting of the single bag $V(G)$, which has width at most $t-1$.
If $G$ has no cycle or vertex of degree at least $3$, then it is a path.
In that case, we set $h(G,b)=0$ and $W(G,b)=\witness{V(G)}$, and we let $P(G,b)$ be any path decomposition of $G$ of width $1$.

Assume henceforth that $T$ has more than one node.
For each node $\alpha\in\Sub(G)\setminus\{G\}$, we run $\solve(\alpha,b)$ to compute $h(\alpha,b)$, $W(\alpha,b)$, and $P(\alpha,b)$ when $h(\alpha,b)<b$.
We call these recursive calls to $\solve$ \emph{primary}.
If any of them leads to $h(\alpha,b)=b$, we just set $h(G,b)=b$ and $W(G,b)=W(\alpha,b)$, and we terminate the run of $\solve(G,b)$, again saying that it is \emph{pruned}.

Assume henceforth that the run of $\solve(G,b)$ is not pruned, that is, we have $h(\alpha,b)<b$ for every $\alpha\in\Sub(G)\setminus\{G\}$.
Let $k$ be the maximum value of $h(\alpha,b)$ for $\alpha\in\Sub(G)\setminus\{G\}$.
Thus $k<b$.
We will consider several further cases, each leading to one of the following two outcomes:
\begin{enumerate}[label=(\Alph*),widest=A]
\item We set $h(G,b)=k$.
In that case, we let $W(G,b)$ be $W(\alpha,b)$ for any node $\alpha\in\Sub(G)\setminus\{G\}$ such that $h(\alpha,b)=k$, and we only need to provide an appropriate path decomposition $P(G,b)$.
\item\label{outcomeB} We set $h(G,b)=k+1$.
In that case, if $k+1<b$, we let $P(G,b)$ be the path decomposition of $G$ of width $t(k+1)+1$ obtained by applying Lemma~\ref{lem:combine}~\ref{item:combine-vertex} with $q$ the root of $T$, and we only need to provide an appropriate $\calT_{k+1}$-witness $W(G,b)$.
\end{enumerate}

Let $Z$ be the set of minimal nodes $\zeta\in\Sub(G)\setminus\{G\}$ such that $h(\zeta,b)=k$.
It follows that $Z\neq\emptyset$ (by the definition of $k$) and the sets $V(\zeta)$ with $\zeta\in Z$ are pairwise disjoint and non-adjacent in $G$.

Suppose that $Z$ consists of a single node $\zeta$.
Let $Q$ be the path from the root to $\zeta$ in $T$.
Thus $h(\gamma,b)<k$ for every node $\gamma$ in $T-Q$.
Every connected component $\gamma$ of $G-\bigcup_{\xi\in V(Q)}B_\xi$ needs to satisfy Lemma~\ref{lem:component}~(a), so it is a node of $T-Q$; in particular, $\gamma\in\Sub(G)\setminus\{G\}$ and $h(\gamma,b)<k$, so $P(\gamma,b)$ is a path decomposition of $\gamma$ of width at most $t(k-1)+1$.
We set $h(G,b)=k$ and apply Lemma~\ref{lem:combine}~\ref{item:combine-subpath} to obtain a path decomposition $P(G,b)$ of $G$ of width at most $tk+1$.

Assume henceforth that $\size{Z}\geq 2$.
Let $U=V(G)\setminus\bigcup_{\zeta\in Z}V(\zeta)$.
A \emph{$U$-path} is a path in $G$ with all internal vertices in $U$.
Consider an auxiliary graph $H$ with vertex set $Z$ where $\zeta\xi$ is an edge if and only if there is a $U$-path connecting $V(\zeta)$ and $V(\xi)$.
The graph $H$ is connected, because so is $G$.

Suppose that $H$ has a cycle or a vertex of degree at least $3$.
Then there are $\zeta_1,\zeta_2,\zeta_3\in Z$ such that any two of them can be connected by a path in $H$ avoiding the third one.
This and connectedness of the induced subgraphs in $Z$ imply that any two of the sets $V(\zeta_1),V(\zeta_2),V(\zeta_3)$ can be connected by a path in $G$ avoiding the third one.
We set $h(G,b)=k+1$ and $W(G,b)=\witness{V(G);W(\zeta_1,b),W(\zeta_2,b),W(\zeta_3,b)}$, while $P(G,b)$ is constructed as in the description of outcome~\ref{outcomeB}.

Assume henceforth that $H$ has no cycle or vertex of degree at least $3$.
The run of $\solve(G,b)$ is called a \emph{key run} if this case is reached.
It follows that $H$ is a path $\zeta_1\ldots\zeta_m$ with $m=\size{Z}\geq 2$, and every vertex in $U$ is connected by a $U$-path to only one set or to two consecutive sets among $V(\zeta_1),\ldots,V(\zeta_m)$.
Define subsets $U_{0,1},U_{1,2},\ldots,U_{m,m+1}$ of $U$ as follows:
\begin{itemize}
\item $U_{0,1}$ is the set of vertices in $U$ connected by a $U$-path to $V(\zeta_1)$ but not to $V(\zeta_2)$;
\item for $1\leq i<m$, $U_{i,i+1}$ is the set of vertices in $U$ connected by a $U$-path to $V(\zeta_i)$ and to $V(\zeta_{i+1})$;
\item $U_{m,m+1}$ is the set of vertices in $U$ connected by a $U$-path to $V(\zeta_m)$ but not to $V(\zeta_{m-1})$.
\end{itemize}
For $1<i<m$, let $U_i$ be the set of vertices in $U\setminus\bigcup_{j=0}^mU_{j,j+1}$ connected by a $U$-path to $V(\zeta_i)$.
The sets $U_i$ ($1<i<m$) and $U_{i,i+1}$ ($0\leq i\leq m$) are pairwise disjoint, and their union is $U$.
Let $B_i=B_{\zeta_i}$ for $1\leq i\leq m$.
All possible intersections and adjacencies between the sets $V(\zeta_i)$, $U_{i,i+1}$, $U_i$, and $B_i$ are illustrated in the following diagram by overlaps and touchings:

\begin{center}
\begin{tikzpicture}[scale=.68]
\def\V#1{(5*#1,0) arc (180:0:1)--+(0,-1.4) arc (0:-180:1)--cycle}
\def\W#1{(5*#1+3.5,0) ellipse (2 and 1.2)}
\def\U#1{(5*#1+1,1.5) ellipse (0.7 and 1.2)}
\def\B#1{(5*#1+1,0.2) ellipse (2.6 and 1.1)}
\draw[dotted]\B{1}\B{2}\B{3}\B{4};
\begin{scope}
  \clip\W{0}\W{1}\U{2}\W{2}\U{3}\W{3}\W{4};
  \draw[thick,black!30,fill=black!10]\B{1}\B{2}\B{3}\B{4};
\end{scope}
\draw\W{0}\W{1}\U{2}\W{2}\U{3}\W{3}\W{4};
\begin{scope}
  \clip\V{1}\V{2}\V{3}\V{4};
  \draw[thick,black!30,fill=black!10]\B{1}\B{2}\B{3}\B{4};
\end{scope}
\draw\V{1}\V{2}\V{3}\V{4};
\node at (6,-1.55) {$V(\zeta_1)$};
\node at (11,-1.55) {$V(\zeta_2)$};
\node at (16,-1.55) {$V(\zeta_3)$};
\node at (21,-1.55) {$V(\zeta_4)$};
\node at (2.5,-0.05) {$U_{0,1}$};
\node at (8.5,-1.65) {$U_{1,2}$};
\node at (13.5,-1.65) {$U_{2,3}$};
\node at (18.5,-1.65) {$U_{3,4}$};
\node at (24.5,-0.05) {$U_{4,5}$};
\node at (11,1.9) {$U_2$};
\node at (16,1.9) {$U_3$};
\node[black!70] at (6,0) {$B_1$};
\node[black!70] at (11,0) {$B_2$};
\node[black!70] at (16,0) {$B_3$};
\node[black!70] at (21,0) {$B_4$};
\end{tikzpicture}
\end{center}

\noindent
Observe that $B_i\subseteq V(\zeta_i)\cup U_{i-1,i}\cup U_i\cup U_{i,i+1}$ for $1<i<m$ and $B_i\subseteq V(\zeta_i)\cup U_{i-1,i}\cup U_{i,i+1}$ for $i=1$ and $i=m$.
This is because $B_i\subseteq V(\zeta_i)\cup N_G(V(\zeta_i))$, by the definition of $\{B_\alpha\}_{\alpha\in\Sub(G)}$.

It may seem peculiar that we define sets $U_{0,1}$ and $U_{m,m+1}$ rather than $U_1$ and $U_m$, so here is some intuition.
The sequence of bags $B_1,\ldots,B_m$ with appropriate connections inside the sets $U_{1,2},\ldots,U_{m-1,m}$ provides a skeleton along which a path decomposition of $G$ can be laid out.
The sets $U_{0,1}$ and $U_{m,m+1}$ attach to the ends of this skeleton, so we can as well extend it inside these sets (and we do---towards the root of $T$), making them behave similarly to $U_{1,2},\ldots,U_{m-1,m}$.
Furthermore, a path decomposition of $G$ needs to incorporate a path decomposition of each subgraph $G[U_i]$ ($1<i<m$) of appropriately smaller width.
The argument guarantees that if we cannot provide the latter, then $G[U_i]$ contains an appropriate witness which we can combine with witnesses for $\zeta_{i-1}$ and $\zeta_{i+1}$ into a larger witness for $G$.
We would lack one of the witnesses if we tried the same for $G[U_1]$ or $G[U_m]$ had they been defined.
Having explained the intuition, we proceed with the argument.

Let $V_1=V(\zeta_1)\cup B_1$, $V_i=V(\zeta_i)\cup B_i\cup U_i$ for $1<i<m$, and $V_m=V(\zeta_m)\cup B_m$.
Let $Q_{0,1}$ be the path in $T$ from the root $G$ to $\zeta_1$, let $Q_{i,i+1}$ be the path in $T$ from $\zeta_i$ to $\zeta_{i+1}$ for $1\leq i<m$, and let $Q_{m,m+1}$ be the path in $T$ from $\zeta_m$ to the root $G$.
Let $B_{Q_{i,i+1}}=\bigcup_{\xi\in V(Q_{i,i+1})}B_\xi$ for $0\leq i\leq m$.

Let $\Gamma$ be the family of connected components of $G[V_i\setminus B_i]$ for $1\leq i\leq m$ and of connected components of $G[U_{i,i+1}\setminus B_{Q_{i,i+1}}]$ for $0\leq i\leq m$.
Suppose that for each $\gamma\in\Gamma$, we have a path decomposition $P_\gamma$ of $\gamma$ of width at most $t(k-1)+1$.
Then we apply
\begin{itemize}
\item Lemma~\ref{lem:combine}~\ref{item:combine-vertex} to $G[V_i]$, the tree decomposition $(T,\{V_i\cap B_\alpha\}_{\alpha\in\Sub(G)})$ of $G[V_i]$ and the node $\zeta_i$ to obtain a path decomposition $P_i$ of $G[V_i]$ of width at most $tk+1$ containing the set $B_i$ in every bag, for $1\leq i\leq m$;
\item Lemma~\ref{lem:combine}~\ref{item:combine-subpath} to $G[U_{i,i+1}]$, the tree decomposition $(T,\{U_{i,i+1}\cap B_\alpha\}_{\alpha\in\Sub(G)})$ of $G[U_{i,i+1}]$ and the path $Q_{i,i+1}$ to obtain a path decomposition $P_{i,i+1}$ of $G[U_{i,i+1}]$ of width at most $tk+1$, for $0\leq i\leq m$; moreover, Lemma~\ref{lem:combine}~\ref{item:combine-subpath} guarantees that the path decomposition $P_{i,i+1}$ contains $U_{i,i+1}\cap B_i$ in the first bag if $i\geq 1$ and $U_{i,i+1}\cap B_{i+1}$ in the last bag if $i<m$.
\end{itemize}
We set $h(G,b)=k$, and we concatenate these path decompositions $P_{0,1},P_1,P_{1,2},\ldots,P_m,P_{m,m+1}$ to obtain a path decomposition $P(G,b)$ of $G$ of width at most $tk+1$, as shown by the following claim.

\begin{claim}
The concatenation of\/ $P_{0,1},P_1,P_{1,2},\ldots,P_m,P_{m,m+1}$ is a path decomposition of\/ $G$.
\end{claim}

\begin{proof}
If $v\in V(\zeta_i)$ with $1\leq i\leq m$ or $v\in U_i$ with $1<i<m$, then $v$ belongs to bags in $P_i$, where the corresponding nodes form a non-empty subpath, and it belongs to no bags in the other path decompositions among $P_{0,1},P_1,P_{1,2},\ldots,P_m,P_{m,m+1}$.
If $v\in U_{i,i+1}$ with $0\leq i\leq m$, then $v$ belongs to bags in $P_{i,i+1}$, where the corresponding nodes form a non-empty subpath, and moreover,
\begin{itemize}
\item if $v\in B_i\cap U_{i,i+1}$ with $1\leq i\leq m$, then $v$ belongs to all bags of $P_i$ and to the first bag of $P_{i,i+1}$;
\item if $v\in U_{i,i+1}\cap B_{i+1}$ with $0\leq i<m$, then $v$ belongs to the last bag of $P_{i,i+1}$ and to all bags of $P_{i+1}$;
\item $v$ belongs to no bags in the other path decompositions among $P_{0,1},P_1,P_{1,2},\ldots,P_m,P_{m,m+1}$.
\end{itemize}
In all cases, the nodes whose bags contain $v$ form a non-empty subpath in the concatenation.

Now, consider an edge $uv$ of $G$.
If $u,v\in V_i$ ($1\leq i\leq m$) or $u,v\in U_{i,i+1}$ ($0\leq i\leq m$), then both $u$ and $v$ belong to a common bag of $P_i$ or $P_{i,i+1}$ (respectively).
Since the sets of the form $U_i$ and $U_{i,i+1}$ are pairwise non-adjacent, if the edge $uv$ connects two different sets of the form $U_{i,i+1}$, $U_i$, or $V(\zeta_i)$, then it connects $V(\zeta_i)$ with $N_G(V(\zeta_i))$ ($1\leq i\leq m$), where the latter set is contained in $B_i$ (by definition), so $u,v\in V_i$.
Therefore, the above exhausts all the edges of $G$, showing that every edge is realized in some bag of the concatenation.
\end{proof}

It remains to provide the path decompositions $P_\gamma$ for all $\gamma\in\Gamma$ or to deal with the cases where we cannot provide them.
Let $K$ be the unique subtree of $T$ that has $G$ as the root and the nodes in $Z$ as the leaves.
Thus $h(\gamma,b)<k$ for every node $\gamma$ in $T-K$, by the definition of $Z$.
For $0\leq i\leq m$, let $\zeta_{i,i+1}$ be the root of $Q_{i,i+1}$ in $T$, that is, the lowest common ancestor of $\zeta_i$ and $\zeta_{i+1}$ in $T$.

\begin{claim}
\label{claim:component}
If\/ $\gamma$ is a connected component of\/ $G[V_i\setminus B_i]$ where\/ $1\leq i\leq m$, then either
\begin{enumalph}
\item $\gamma$ is a node in\/ $T-K$ that is a child of\/ $\zeta_i$ in\/ $T$, or
\item $V(\gamma)\cap V(\zeta_i)=\emptyset$ and\/ $\gamma\in\CC{G}(U_i\setminus B_i)$, which is possible only when\/ $1<i<m$.
\end{enumalph}
If\/ $\gamma$ is a connected component of\/ $G[U_{i,i+1}\setminus B_{Q_{i,i+1}}]$ where\/ $0\leq i\leq m$, then either
\begin{enumalph}[resume]
\item $\gamma$ is a node in\/ $T-K$ that is a child in\/ $T$ of a node from\/ $Q_{i,i+1}$, or
\item $V(\gamma)\cap V(\zeta_{i,i+1})=\emptyset$, which is possible only when\/ $1\leq i<m$.
\end{enumalph}
\end{claim}

\begin{proof}
For the proof of the first statement, let $\gamma$ be a connected component of $G[V_i\setminus B_i]$ where $1\leq i\leq m$.
Since $N_G(V(\zeta_i))\subseteq B_i$ and $N_G(U_i)\subseteq V(\zeta_i)$ (if $1<i<m$), we have $N_G(V(\zeta_i)\setminus B_i)\subseteq B_i$ and $N_G(U_i\setminus B_i)\subseteq B_i$ (if $1<i<m$).
Thus $N_G(V_i\setminus B_i)\subseteq B_i$, which implies that $\gamma$ is a connected component of $G-B_i$.
By Lemma~\ref{lem:component} applied with $Q$ consisting of the single node $\zeta_i$, either
\begin{enumalph}
\item $\gamma$ is a child of $\zeta_i$ in $T$, so it is a node in $T-K$, as $\zeta_i$ is a leaf of $K$, or
\item $V(\gamma)\cap V(\zeta_i)=\emptyset$, which is possible only when $1<i<m$, because $V_1=V(\zeta_1)\cup B_1$ and $V_m=V(\zeta_m)\cup B_m$; in that case, $V(\gamma)\subseteq U_i\setminus B_i$, so $\gamma\in\CC{G}(U_i\setminus B_i)$, as $N_G(U_i\setminus B_i)\subseteq B_i$.
\end{enumalph}

For the proof of the second statement, let $\gamma$ be a connected component of $G[U_{i,i+1}\setminus B_{Q_{i,i+1}}]$ where~$0\leq i\leq m$.
We have $N_G(U_{i,i+1})\subseteq V(\zeta_i)\cup V(\zeta_{i+1})$, $N_G(V(\zeta_i))\subseteq B_i$, and $N_G(V(\zeta_{i+1}))\subseteq B_{i+1}$, which imply $N_G(U_{i,i+1}\setminus B_{Q_{i,i+1}})\subseteq B_{Q_{i,i+1}}$, as $\zeta_i,\zeta_{i+1}\in V(Q_{i,i+1})$.
Therefore, $\gamma$ is a connected component of $G-B_{Q_{i,i+1}}$.
By Lemma~\ref{lem:component} applied with $Q=Q_{i,i+1}$, either
\begin{enumalph}[resume]
\item $\gamma$ is a node in $T-Q_{i,i+1}$ that is a child in $T$ of some node $\xi$ of $Q_{i,i+1}$, so $\gamma$ is a node in $T-K$, as $V(\gamma)$ is disjoint from $V(\zeta_j)$ for every leaf $\zeta_j$ of $K$, or
\item $V(\gamma)\cap V(\zeta_{i,i+1})=\emptyset$, which is possible only when $1\leq i<m$ (otherwise $\zeta_{i,i+1}=G$).\qedhere
\end{enumalph}
\end{proof}

Let a component $\gamma\in\Gamma$ be called a \emph{child component} when case (a) or~(c) of Claim~\ref{claim:component} holds for $\gamma$ and a \emph{parent component} when case (b) or~(d) of Claim~\ref{claim:component} holds for $\gamma$.
Case (a) or~(c) of Claim~\ref{claim:component} implies that for every child component $\gamma$, there has been a primary recursive call to $\solve(\gamma,b)$ and $h(\gamma,b)<k$, so $P_\gamma=P(\gamma,b)$ is a path decomposition of $\gamma$ of width at most $t(k-1)-1$.

An attempt to deal with parent components $\gamma$ would be to run $\solve(\gamma,k)$ for each of them to compute $h(\gamma,k)$, $W(\gamma,k)$, and $P(\gamma,k)$ when $h(\gamma,k)<k$.
If every such recursive call led to $h(\gamma,k)<k$, then $P(\gamma,k)$ would be a requested path decomposition $P_\gamma$ of $\gamma$ of width at most $t(k-1)+1$ for every parent component $\gamma$.
If some of these recursive calls led to $h(\gamma,k)=k$, then we would set $h(G,b)=k+1$ and use $W(\gamma,k)$ to construct a requested $\calT_{k+1}$-witness $W(G,b)$, while $P(G,b)$ would be constructed as described in outcome~\ref{outcomeB} with no need for explicit path decompositions of the parent components.
However, this approach fails in that we are unable to provide a polynomial upper bound on the number of distinct pairs $(G,b)$ for which recursive calls to $\solve(G,b)$ would be made.

We overcome this difficulty as follows.
For each parent component $\gamma$, instead of running $\solve(\gamma,k)$, we run $\solve(\hat\gamma,k)$ for an appropriately defined connected induced subgraph $\hat\gamma$ of $G[U]$ such that $V(\gamma)\subseteq V(\hat\gamma)$, to compute $h(\hat\gamma,k)$, $W(\hat\gamma,k)$, and $P(\hat\gamma,k)$ when $h(\hat\gamma,k)<k$.
We call these recursive calls \emph{secondary}.
Their role is analogous to the role of the calls to $\solve(\gamma,k)$ in the attempt above.
If every secondary call leads to $h(\hat\gamma,k)<k$, then $P(\hat\gamma,k)$ is a path decomposition of $\hat\gamma$ of width at most $t(k-1)+1$, and the requested path decomposition $P_\gamma$ of every parent component $\gamma$ is obtained from the respective $P(\hat\gamma,k)$ by restricting the bags to $V(\gamma)$.
If some secondary call leads to $h(\hat\gamma,k)=k$, then we set $h(G,b)=k+1$ and use $W(\hat\gamma,k)$ to construct a requested $\calT_{k+1}$-witness $W(G,b)$ (as described in Claim~\ref{claim:witness} below), while $P(G,b)$ is constructed as described in outcome~\ref{outcomeB} with no need for explicit path decompositions of the parent components.
The induced subgraphs $\hat\gamma$ are defined in a way (described below) that behaves nicely in the recursion and will allow us (in Subsection~\ref{subsec:complexity}) to provide a polynomial upper bound on the number of distinct pairs $(G,b)$ for which recursive calls to $\solve(G,b)$ are made.

The definition that follows is technical, but it is exactly what we need to be able to prove Lemma~\ref{lem:level}.
For $\sigma\in\Sub(G)$, let $A_\sigma$ be as in the definition of normalized tree decomposition, so that $A_\sigma\subseteq V(\sigma)$ and $B_\sigma=A_\sigma\cup N_G(V(\sigma))$.
Observe that if $v\in B_\xi$ where $\xi\in\Sub(G)$, then all nodes in $\Sub(G)$ containing $v$, in particular the node $\sigma$ such that $v\in A_\sigma$, lie on the path from $\xi$ to the root $G$ in $T$.
For $1\leq i<m$, let $R_{i,i+1}$ be the set of vertices $v\in U\cap B_{\zeta_{i,i+1}}$ with the following property: if $\sigma$ is the node in $\Sub(G)$ such that $v\in A_\sigma$ (which therefore lies on the path from $\zeta_{i,i+1}$ to the root $G$ in $T$), then $v$ is connected by a $(U\cap V(\sigma))$-path to $V(\zeta_i)$ and to $V(\zeta_{i+1})$; thus $R_{i,i+1}\subseteq U_{i,i+1}$.
For $1\leq i\leq m$, let $R_i=U\cap B_i=B_i\setminus V(\zeta_i)$; thus $R_1\subseteq U_{0,1}\cup U_{1,2}$, $R_i\subseteq U_{i-1,i}\cup U_i\cup U_{i,i+1}$ if $1<i<m$, and $R_m\subseteq U_{m-1,m}\cup U_{m,m+1}$.
Let $R=R_1\cup R_{1,2}\cup R_2\cup\cdots\cup R_{m-1,m}\cup R_m$.

For each parent component $\gamma$, the induced subgraph $\hat\gamma$ used in the description above is defined as follows.
If $\gamma\in\CC{G}(U_i\setminus B_i)$ with $1<i<m$, as in case~(b) of Claim~\ref{claim:component}, then $\hat\gamma=\gamma$, which is a connected component of $G[U_i\setminus R]$, because $U_i\cap R=U_i\cap R_i=U_i\cap B_i$.
If $\gamma\in\CC{G}(U_{i,i+1}\setminus B_{Q_{i,i+1}})$ with $1\leq i<m$, as in case~(d) of Claim~\ref{claim:component}, then $\hat\gamma$ is the connected component of $G[U_{i,i+1}\setminus R]$ such that $V(\gamma)\subseteq V(\hat\gamma)$, which exists because $U_{i,i+1}\cap R=U_{i,i+1}\cap(R_i\cup R_{i,i+1}\cup R_{i+1})\subseteq B_{Q_{i,i+1}}$.
Let each $\hat\gamma$ obtained this way from a parent component $\gamma$ be called a \emph{secondary component}.
Every secondary component is a connected component of $G[U\setminus R]$ and thus of $G-R$, as $N_G(\bigcup_{i=1}^mV(\zeta_i))\subseteq R$.

\begin{claim}
\label{claim:witness}
If\/ $h(\hat\gamma,k)=k$ for a secondary component\/ $\hat\gamma$, then the following is a\/ $\calT_{k+1}$-witness for\/ $G$:
\begin{itemize}
\item $\witness{V(G);W(\zeta_{i-1},b),W(\hat\gamma,k),W(\zeta_{i+1},b)}$ when\/ $\hat\gamma\in\CC{G}(U_i\setminus R)$ with\/ $1<i<m$;
\item
$\witness{V(G);W(\zeta_i,b),W(\hat\gamma,k),W(\zeta_{i+1},b)}$ when\/ $\hat\gamma\in\CC{G}(U_{i,i+1}\setminus R)$ with\/ $1\leq i<m$.
\end{itemize}
\end{claim}

\begin{proof}
If $\hat\gamma\in\CC{G}(U_i\setminus R)$ with $1<i<m$, then there is a path in $G$ connecting
\begin{itemize}
\item $V(\zeta_{i-1})$ with $V(\hat\gamma)$ via $U_{i-1,i}\cup V_i$, thus avoiding $V(\zeta_{i+1})$,
\item $V(\hat\gamma)$ with $V(\zeta_{i+1})$ via $V_i\cup U_{i,i+1}$, thus avoiding $V(\zeta_{i-1})$,
\item $V(\zeta_{i-1})$ with $V(\zeta_{i+1})$ via $U_{i-1,i}\cup V(\zeta_i)\cup U_{i,i+1}$, thus avoiding $V(\hat\gamma)$.
\end{itemize}
This shows that $\witness{V(G);W(\zeta_{i-1},b),W(\hat\gamma,k),W(\zeta_{i+1},b)}$ is a $\calT_{k+1}$-witness for $G$.

Now, suppose $\hat\gamma\in\CC{G}(U_{i,i+1}\setminus R)$ with $1\leq i<m$.
There is a path in $G$ connecting
\begin{itemize}
\item $V(\zeta_i)$ with $V(\hat\gamma)$ via $U_{i,i+1}$, thus avoiding $V(\zeta_{i+1})$,
\item $V(\hat\gamma)$ with $V(\zeta_{i+1})$ via $U_{i,i+1}$, thus avoiding $V(\zeta_i)$.
\end{itemize}
To show that $\witness{V(G);W(\zeta_i,b),W(\hat\gamma,k),W(\zeta_{i+1},b)}$ is a $\calT_{k+1}$-witness for $G$, it remains to provide a path in $G$ connecting $V(\zeta_i)$ with $V(\zeta_{i+1})$ that avoids $V(\hat\gamma)$.

Let $X$ be the set of vertices in $U\cap V(\zeta_{i,i+1})$ that are connected by a $(U\cap V(\zeta_{i,i+1}))$-path to $V(\zeta_i)$ and to $V(\zeta_{i+1})$.
Since $\zeta_{i,i+1}$ is connected and $V(\zeta_i)\cup V(\zeta_{i+1})\subseteq V(\zeta_{i,i+1})$, there is a $V(\zeta_{i,i+1})$-path connecting $V(\zeta_i)$ and $V(\zeta_{i+1})$.
A minimal path with that property has all internal vertices in $U$ and therefore in $X$.
To complete the proof, we show that $V(\hat\gamma)\cap X=\emptyset$.

Suppose for the sake of contradiction that $V(\hat\gamma)\cap X\neq\emptyset$.
Since $\hat\gamma$ is a secondary component, there is a parent component $\gamma\in\CC{G}(U_{i,i+1}\setminus B_{Q_{i,i+1}})$ such that $V(\gamma)\cap V(\zeta_{i,i+1})=\emptyset$ and $V(\gamma)\subseteq V(\hat\gamma)$, so $V(\hat\gamma)\nsubseteq X$.
Therefore, since $\hat\gamma$ is connected, there is a vertex $v\in V(\hat\gamma)\setminus X$ with a neighbor in $X$.
It follows that $v$ is connected by a $(U\cap V(\zeta_{i,i+1}))$-path to $V(\zeta_i)$ and to $V(\zeta_{i+1})$.
Since $v\in U\setminus X$, we have $v\notin V(\zeta_{i,i+1})$ (by the definition of $X$) and thus $v\in N_G(V(\zeta_{i,i+1}))\subseteq B_{\zeta_{i,i+1}}$.
Let $\sigma$ be the node on the path from $\zeta_{i,i+1}$ to the root in $T$ such that $v\in A_\sigma$.
Thus $V(\zeta_{i,i+1})\subseteq V(\sigma)$.
We conclude that $v\in U\cap B_{\zeta_{i,i+1}}$ and $v$ is connected by a $(U\cap V(\sigma))$-path to $V(\zeta_i)$ and to $V(\zeta_{i+1})$, so $v\in R_{i,i+1}$, which contradicts the fact that $V(\hat\gamma)\cap R=\emptyset$.
\end{proof}

This completes the description of the procedure $\solve(G,b)$.
Since all recursive calls of the form $\solve(\gamma,c)$ that it makes are for proper connected induced subgraphs $\gamma$ of $G$ (and for $c\leq b$), the procedure terminates and correctly computes the requested outcome.

\subsection{Complexity Analysis}
\label{subsec:complexity}

The algorithm consists in running $\solve(\inputG,\infty)$ on the input graph $\inputG$.
It makes further recursive calls to $\solve(\beta,b)$ for various connected induced subgraphs $\beta$ of $\inputG$ and upper bound requests $b$.
Let every pair $(\beta,b)$ such that $\solve(\beta,b)$ is run by the algorithm (somewhere in the recursion tree) be called a \emph{subproblem}.
We show that if $\inputG$ has $n$ vertices, then there are only $O(n\log n)$ subproblems.

A \emph{pruned subproblem} is a subproblem $(\beta,b)$ for which the run of $\solve(\beta,b)$ is pruned, which implies $h(\beta,b)=b$.
Observe that unless $(\beta,b)$ is pruned, the run of $\solve(\beta,b)$ performs operations that do not depend on the value of $b$ (including the operations performed in all recursive calls).
Indeed, unless $(\beta,b)$ is pruned, no primary recursive call made by $\solve(\beta,b)$ is pruned, so (by induction) these calls perform operations and lead to results that do not depend on $b$, and every secondary recursive call made by $\solve(\beta,b)$ is of the form $\solve(\hat\gamma,k)$ where $k$ does not depend on $b$.

A \emph{key subproblem} is a subproblem $(\beta,b)$ for which the run of $\solve(\beta,b)$ is a key run, that is, $\size{Z}\geq 2$ and $H$ is a path.
In particular, a key subproblem is not pruned.
A \emph{key subgraph} is a connected induced subgraph $\beta$ of $\inputG$ such that $(\beta,b)$ is a key subproblem for some $b\in\setN\cup\{\infty\}$.

For every key subgraph $\beta$, let $\level(\beta)$ denote the value of $k$ defined in a key run of $\solve(\beta,b)$, that is, the maximum of $h(\gamma,b)$ over all $\gamma\in\Sub(\beta)\setminus\{\beta\}$.
(As we observed above, these values do not depend on $b$, so neither does $\level(\beta)$.)
For every key subgraph $\beta$, let $R(\beta)$ be the set $R=R_1\cup R_{1,2}\cup R_2\cup\cdots\cup R_{m-1,m}\cup R_m$ defined in a key run of $\solve(\beta,b)$.
(As before, it does not depend on $b$.)
The following lemma exhibits the crucial property of the sets $R(\beta)$.

\begin{lemma}
\label{lem:level}
The following holds for any key subgraphs\/ $\alpha$ and\/ $\beta$ such that\/ $\alpha\in\Sub(\beta)$.
\begin{enumerate}
\item\label{item:level-order} $\level(\alpha)\leq\level(\beta)$.
\item\label{item:level-less} If\/ $\level(\alpha)<\level(\beta)$, then\/ $R(\beta)\cap V(\alpha)=\emptyset$.
\item\label{item:level-equal} If\/ $\level(\alpha)=\level(\beta)$, then\/ $R(\beta)\cap V(\alpha)=R(\alpha)$.
\end{enumerate}
\end{lemma}

\begin{proof}
Since $\alpha\in\Sub(\beta)$, it follows from Lemma~\ref{lem:Tin}~\ref{item:Tin-recursion} that $\Sub(\alpha)\subseteq\Sub(\beta)$.
This directly shows~\ref{item:level-order}, by the definition of level.

Let $(T,\{B_\sigma\}_{\sigma\in\Sub(\beta)})$ be the normalized tree decomposition of $\beta$ used in a key run of $\solve$ on $\beta$, and let $\{A_\sigma\}_{\sigma\in\Sub(\beta)}$ be the corresponding sets from the definition of normalized tree decomposition.
(Their construction from the input tree decomposition of $\inputG$ depends only on $\beta$.)
Lemma~\ref{lem:Tin}~\ref{item:Tin-recursion} implies that $(T^\alpha,\{B_\sigma\cap V(\alpha)\}_{\sigma\in\Sub(\alpha)})$, where $T^\alpha$ is the subtree of $T$ comprising $\alpha$ and the nodes below $\alpha$ in $T$, is the normalized tree decomposition of $\alpha$ used in any key run of $\solve$ on $\alpha$.
Let $Z$, $m$, $H$, $U$, $\zeta_i$, $\zeta_{i,i+1}$, $B_i$, $R_i$, and $R_{i,i+1}$ be defined as in a key run of $\solve$ on $\beta$.

For the proof of~\ref{item:level-less}, assume $\level(\alpha)<\level(\beta)$.
The definition of level implies that none of the nodes in $Z$ lie below $\alpha$ in $T$.
Recall that if $v\in B_\xi$ where $\xi\in\Sub(\beta)$, then all nodes in $\Sub(\beta)$ containing $v$ lie on the path from $\xi$ to the root $\beta$ in $T$.
Therefore, for $1\leq i\leq m$, we have $R_i\cap V(\alpha)=\emptyset$, as $R_i=B_i\setminus V(\zeta_i)$.
Similarly, for $1\leq i<m$, we have $R_{i,i+1}\cap V(\alpha)=\emptyset$, as $R_{i,i+1}\subseteq B_{\zeta_{i,i+1}}$, where the node $\zeta_{i,i+1}$ lies above $\zeta_i$ and $\zeta_{i+1}$ in $T$.
This yields $R(\beta)\cap V(\alpha)=\emptyset$.

For the proof of~\ref{item:level-equal}, assume $\level(\alpha)=\level(\beta)$.
The definition of level implies that $\alpha\notin Z$ and none of the nodes in $Z$ lies above $\alpha$ in $T$.
Let $Z^\alpha$, $H^\alpha$, and $U^\alpha$ be the respective $Z$, $H$, and $U$ defined in any key run of $\solve$ on $\alpha$.
(We keep using $Z$, $H$, $U$, and other notations with no superscript to denote what is defined in a key run of $\solve$ on $\beta$.)
It follows that $Z^\alpha$ is the set of nodes in $Z$ that lie below $\alpha$ in $T$.
Lemma~\ref{lem:Sub}~\ref{item:Sub-laminar} implies $V(\zeta_i)\cap V(\alpha)=\emptyset$ for every $\zeta_i\in Z\setminus Z^\alpha$.
It follows that $U^\alpha=U\cap V(\alpha)$.
Consequently, the path $H^\alpha$ is a subpath of $H$, and $Z^\alpha=\{\zeta_r,\ldots,\zeta_s\}$ where $1\leq r<s\leq m$.
Let $R^\alpha_r,R^\alpha_{r,r+1},R^\alpha_{r+1},\ldots,R^\alpha_{s-1,s},R^\alpha_s$ denote the sets $R_1,R_{1,2},R_2,\ldots,R_{m-1,m},R_m$ defined in a key run of $\solve$ on $\alpha$ in this or the reverse order matching the order of $\zeta_r,\ldots,\zeta_s$.

Let $1\leq i\leq m$.
If $\zeta_i\in Z^\alpha$, then $R^\alpha_i=U^\alpha\cap B_i\cap V(\alpha)=U\cap B_i\cap V(\alpha)=R_i\cap V(\alpha)$.
If $\zeta_i\notin Z^\alpha$, then $\zeta_i$ does not lie below $\alpha$ in $T$, so $R_i\cap V(\alpha)=(B_i\setminus V(\zeta_i))\cap V(\alpha)=\emptyset$.
Now, let $1\leq i<m$.
Suppose $\zeta_i,\zeta_{i+1}\in Z^\alpha$.
It follows that $\zeta_{i,i+1}$ is a node in $T^\alpha$.
For each node $\sigma$ on the path from $\zeta_{i,i+1}$ to $\alpha$ in $T$, since $A_\sigma\subseteq V(\sigma)\subseteq V(\alpha)$, the set $A_\sigma$ contributes the same vertices to both $R^\alpha_{i,i+1}$ and $R_{i,i+1}$, namely, the vertices in $U\cap B_{\zeta_{i,i+1}}\cap A_\sigma$ that are connected by a $(U\cap V(\sigma))$-path to $V(\zeta_i)$ and to $V(\zeta_{i+1})$.
For each node $\sigma$ above $\alpha$ in $T$, the set $A_\sigma$ contributes no vertices to $R^\alpha_{i,i+1}$, and the vertices it contributes to $R_{i,i+1}$ are not in $V(\alpha)$, as $A_\sigma\cap V(\alpha)=\emptyset$.
Thus $R^\alpha_{i,i+1}=R_{i,i+1}\cap V(\alpha)$.
On the other hand, if $\zeta_i\notin Z^\alpha$ or $\zeta_{i+1}\notin Z^\alpha$, then $\zeta_{i,i+1}$ does not lie in $T_\alpha$, so $B_{\zeta_{i,i+1}}\cap V(\alpha)=\emptyset$ and thus $R_{i,i+1}\cap V(\alpha)=\emptyset$.
We conclude that
\[\begin{split}
R(\alpha)&=\smash[b]{R^\alpha_r\cup R^\alpha_{r,r+1}\cup R^\alpha_{r+1}\cup\cdots\cup R^\alpha_{s-1,s}\cup R^\alpha_s}\\
&=(\smash[b]{R_1\cup R_{1,2}\cup R_2\cup\cdots\cup R_{m-1,m}\cup R_m})\cap V(\alpha)\\
&=R(\beta)\cap V(\alpha).\qedhere
\end{split}\]
\end{proof}

Let $\ell$ be the maximum of $\level(\beta)$ over all key subgraphs $\beta$.
It follows that every key subproblem is of the form $(\beta,k)$ with $k\in\{0,\ldots,\ell,\infty\}$.
For $k\in\{0,\ldots,\ell\}$, let $R^k$ be the union of the sets $R(\beta)$ over all key subgraphs $\beta$ such that $\level(\beta)=k$.
(We have $R^k=\emptyset$ if there are no such key subgraphs.)
Let $R^{\geq k}$ be $R^k\cup\cdots\cup R^\ell$ for $k\in\{0,\ldots,\ell\}$ and empty for $k=\infty$.
The following lemma will finally allow us to bound the total number of subproblems.

\begin{lemma}
\label{lem:subproblems}
For every\/ $k\in\{0,\ldots,\ell,\infty\}$, every subproblem\/ $(\gamma,k)$ satisfies\/ $\gamma\in\Sub(\inputG-R^{\geq k})$.
\end{lemma}

\begin{proof}
We aim to prove that the following two statements, the first of which is the statement of the lemma, hold for all $k\in\{0,\ldots,\ell,\infty\}$.
\begin{enumroman}
\item\label{item:subproblems-1} Every subproblem $(\gamma,k)$ satisfies $\gamma\in\Sub(\inputG-R^{\geq k})$.
\item\label{item:subproblems-2} Every key subproblem $(\alpha,a)$ with $a\geq k>\level(\alpha)$ satisfies $\alpha\in\Sub(\inputG-R^{\geq k})$.
\end{enumroman}

Suppose not, and let $k\in\{0,\ldots,\ell,\infty\}$ be maximum for which \ref{item:subproblems-1} or \ref{item:subproblems-2} fails.
Let $R^{>k}$ be $R^{\geq k+1}$ for $k\in\{0,\ldots,\ell-1\}$ and empty for $k\in\{\ell,\infty\}$.
By maximality of $k$, the following statement holds, which is equivalent to \ref{item:subproblems-2} for $k+1$ if $k<\ell$, equivalent to \ref{item:subproblems-2} for $\infty$ if $k=\ell$, and vacuous if $k=\infty$.
\begin{enumerate}[label={($*$)}]
\item\label{item:subproblems-IH} Every key subproblem $(\alpha,a)$ with $a>k\geq\level(\alpha)$ satisfies $\alpha\in\Sub(\inputG-R^{>k})$.
\end{enumerate}

First, suppose that \ref{item:subproblems-1} fails for $k$.
Consider a subproblem $(\gamma,k)$ with $\gamma$ maximal for which \ref{item:subproblems-1} fails, so that \ref{item:subproblems-1} holds for all subproblems $(\beta,k)$ with $V(\gamma)\subset V(\beta)$.
Since $\inputG\in\Sub(\inputG)$, \ref{item:subproblems-1} holds when $(\gamma,k)=(\inputG,\infty)$.
Thus assume $(\gamma,k)\neq(\inputG,\infty)$, so that a primary or secondary recursive call to $\solve(\gamma,k)$ occurs in the algorithm.
If there is a primary call to $\solve(\gamma,k)$ from $\solve(\beta,k)$ where $\gamma\in\Sub(\beta)\setminus\{\beta\}$, then $\beta\in\Sub(\inputG-R^{\geq k})$ implies $\gamma\in\Sub(\inputG-R^{\geq k})$ by Lemma~\ref{lem:Tin}~\ref{item:Tin-recursion}, which contradicts the assumption that \ref{item:subproblems-1} fails for $(\gamma,k)$.
Thus assume the other case, namely, that $k\leq\ell$ and the algorithm makes a secondary call to $\solve(\gamma,k)$ from $\solve(\alpha,a)$ for some key subproblem $(\alpha,a)$ with $a>k=\level(\alpha)$.
By \ref{item:subproblems-IH}, we have $\alpha\in\Sub(\inputG-R^{>k})$.

We claim that $R^k\cap V(\alpha)=R(\alpha)$.
It is clear that $R^k\cap V(\alpha)\supseteq R(\alpha)$.
Now, let $v\in R^k\cap V(\alpha)$.
Let $(\beta,b)$ be a key subproblem with $b>k=\level(\beta)$ such that $v\in R(\beta)$ (which exists by the definition of $R^k$).
By \ref{item:subproblems-IH}, we have $\beta\in\Sub(\inputG-R^{>k})$.
The fact that $v\in V(\alpha)\cap V(\beta)$ and Lemma~\ref{lem:Tin}~\ref{item:Tin-laminar} yield $\alpha\in\Sub(\beta)$ or $\beta\in\Sub(\alpha)$.
If $\alpha\in\Sub(\beta)$, then Lemma~\ref{lem:level}~\ref{item:level-equal} yields $v\in R(\beta)\cap V(\alpha)=R(\alpha)$.
If $\beta\in\Sub(\alpha)$, then Lemma~\ref{lem:level}~\ref{item:level-equal} (with $\alpha$ and $\beta$ interchanged) yields $v\in R(\beta)\subseteq R(\alpha)$.
This completes the proof of the claim.

Being a secondary component in $\solve(\alpha,a)$, $\gamma$ is a connected component of $\alpha-R(\alpha)$ and thus of $\alpha-R^k$, as $R^k\cap V(\alpha)=R(\alpha)$.
Since $\alpha\in\Sub(\inputG-R^{>k})$, Lemma~\ref{lem:Tin}~\ref{item:Tin-component} applied with $G=\inputG-R^{>k}$ and $X=V(\inputG)\setminus R^{\geq k}$ yields $\gamma\in\Sub(\inputG-R^{\geq k})$, which is again a contradiction.

Having assumed that $k$ is maximal for which \ref{item:subproblems-1} or \ref{item:subproblems-2} fails, we have proved that \ref{item:subproblems-1} actually holds for $k$, so it must be \ref{item:subproblems-2} that fails for $k$.
So let $(\alpha,a)$ be a key subproblem with $a\geq k>\level(\alpha)$.
If $a=k$, then $\alpha\in\Sub(\inputG-R^{\geq k})$ by the fact that \ref{item:subproblems-1} holds for $k$ that we have proved above.
So assume $a>k$.
By \ref{item:subproblems-IH}, we have $\alpha\in\Sub(\inputG-R^{>k})$.
Suppose there is a vertex $v\in R^k\cap V(\alpha)$.
As before, let $(\beta,b)$ be a key subproblem with $b>k=\level(\beta)$ such that $v\in R(\beta)$.
By \ref{item:subproblems-IH}, we have $\beta\in\Sub(\inputG-R^{>k})$.
The fact that $v\in V(\alpha)\cap V(\beta)$ and Lemma~\ref{lem:Tin}~\ref{item:Tin-laminar} yield $\alpha\in\Sub(\beta)$ or $\beta\in\Sub(\alpha)$.
If $\alpha\in\Sub(\beta)$, then Lemma~\ref{lem:level}~\ref{item:level-less} yields $R(\beta)\cap V(\alpha)=\emptyset$, which is a contradiction.
If $\beta\in\Sub(\alpha)$, then $\level(\beta)\leq\level(\alpha)$ by Lemma~\ref{lem:level}~\ref{item:level-order} (with $\alpha$ and $\beta$ interchanged), which is again a contradiction.
Thus $R^k\cap V(\alpha)=\emptyset$, which implies $\alpha\in\Sub(\inputG-R^{\geq k})$ by Lemma~\ref{lem:Tin}~\ref{item:Tin-restrict}.
We conclude that \ref{item:subproblems-2} holds for $k$, which is a final contradiction.
\end{proof}

Let $n=\size{V(\inputG)}$.
Lemma~\ref{lem:normalized}~\ref{item:normalized-A} implies that $\size{\Sub(\inputG[X])}\leq\size{X}\leq n$ for every $X\subseteq V(\inputG)$.
This and Lemma~\ref{lem:subproblems} imply that for every $k\in\{0,\ldots,\ell,\infty\}$, there are at most $n$ subproblems of the form $(\gamma,k)$.
A key subgraph $\beta$ such that $\ell=\level(\beta)$ has a $\calT_\ell$-witness, so $n\geq\size{V(\beta)}\geq 3^\ell$.
Therefore, the total number of subproblems is $O(n\log n)$.
Since the operations performed in a single~pass of $\solve$ (excluding the recursive calls) clearly take polynomial time, we conclude that the full run of $\solve(\inputG,\infty)$ takes polynomial time.
This completes the proof of Theorem~\ref{thm:algorithm}.

\section{Tightness}
\label{sec:tight}

Theorem~\ref{thm1} asserts that every graph with pathwidth at least $th+2$ has treewidth at least $t$ or contains a subdivision of a complete binary tree of height $h+1$.
While this statement is true for all positive integers $t$ and $h$, we remark that the interesting case is when $h>\log_2t-2$.
Indeed, if $h\leq\log_2t-2$, then the second outcome is known to hold for every graph with pathwidth at least $t$; this follows from a result of Bienstock, Robertson, Seymour, and Thomas~\cite{BRST91}.%
\footnote{In~\cite{BRST91}, it is proved that for every forest $F$, graphs with no $F$ minors have pathwidth at most $\size{V(F)}-2$.
In particular, if $G$ contains no subdivision of a complete binary tree of height $h+1$, then $\pw(G)<2^{h+2}\leq t$.}
We now show that Theorem~\ref{thm1} is tight up to a multiplicative factor when $h>\log_2t-2$.

\begin{theorem}
\label{thm:tight}
For any positive integers\/ $t$ and\/ $h$, there is a graph with treewidth\/ $t$ and pathwidth at least\/ $t(h+1)-1$ that contains no subdivision of a complete binary tree of height\/ $3\max(h+1,\lceil\log_2t\rceil)$.
\end{theorem}

Fix a positive integer $t$.
For a tree $T$, let $\blowup{T}$ be a graph obtained from $T$ by replacing every node of $T$ with a clique on $t$ vertices and replacing every edge of $T$ with an arbitrary perfect matching between the corresponding cliques.
For $h\in\setN$, let $T_h$ be a complete ternary tree of height $h$.
The following three claims show that the graph $\blowup{T_h}$ satisfies the three conditions requested in Theorem~\ref{thm:tight}, thus proving that Theorem~\ref{thm:tight} holds for $\blowup{T_h}$.

\begin{claim}
If\/ $T$ is a tree on at least two vertices, then\/ $\tw(\blowup{T})=t$.
\end{claim}

\begin{proof}
For each node $x$ of $T$, let $B_x$ be the clique of $t$ vertices in $\blowup{T}$ corresponding to $x$.
A tree decomposition of $\blowup{T}$ of width $t$ is obtained from $T$ by taking $B_x$ as the bag of every node $x$ of $T$ and by subdividing every edge $xy$ of $T$ into a path of length $t+1$ with the following sequence of bags, assuming that the vertices $u_1,\ldots,u_t$ in $B_x$ are matched to $v_1,\ldots,v_t$ in $B_y$, respectively:
\[\{u_1,\ldots,u_t\},\enspace\{u_1,\ldots,u_t\}\cup\{v_1\},\enspace\{u_2,\ldots,u_t\}\cup\{v_1,v_2\},\enspace\ldots,\enspace\{u_t\}\cup\{v_1,\ldots,v_t\},\enspace\{v_1,\ldots,v_t\}.\]
This way, for every matching edge $u_iv_i$ with $1\leq i\leq t$, there is a bag containing its two endpoints.
Consequently, this is a valid tree decomposition of $\blowup{T}$ with bags of size at most $t+1$.

For the proof that $\tw(\blowup{T})\geq t$, let $xy$ be an edge of $T$, and assume that the vertices $u_1,\ldots,u_t$ of $B_x$ are matched to $v_1,\ldots,v_t$ in $B_y$, respectively, as before.
In any tree decomposition of $\blowup{T}$, there is a node $x'$ whose bag contains the clique $B_x$ and a node $y'$ whose bag contains the clique $B_y$.
Walk on the path from $x'$ to $y'$ and stop at the first node whose bag contains some vertex in $B_y$.
This bag must also contain all of $B_x$, so it has size at least $t+1$.
\end{proof}

\begin{claim}
For every\/ $h\in\setN$, we have\/ $\pw(\blowup{T_h})\geq t(h+1)-1$.
\end{claim}

\begin{proof}
We define the \emph{root clique} of $\blowup{T_h}$ as the clique in $\blowup{T_h}$ corresponding to the root of $T_h$.
We prove the following slightly stronger statement, by induction on $h$: in every path decomposition of $\blowup{T_h}$, there are a bag $B$ of size at least $t(h+1)$ and $t$ vertex-disjoint paths in $\blowup{T_h}$ each having one endpoint in the root clique and the other endpoint in $B$.

For the base case $h=0$, the graph $\blowup{T_0}$ is simply a complete graph on $t$ vertices, and the statement is trivial.
For the induction step, assume that $h\geq 1$ and the statement is true for $h-1$.
Let $R$ be the root clique of $\blowup{T_h}$.
Let $(P,\{B_x\}_{x\in V(P)})$ be a path decomposition of $\blowup{T_h}$ of minimum width.
The graph $\blowup{T_h}-R$ has three connected components $C_1$, $C_2$, and $C_3$ that are copies of $\blowup{T_{h-1}}$ with root cliques $R_1$, $R_2$, and $R_3$, respectively.
For each $i\in\{1,2,3\}$, the induction hypothesis applied to the path decomposition $(P,\{B_x\cap V(C_i)\}_{x\in V(P)})$ of $C_i$ provides a node $x_i$ of $P$ such that
\begin{itemize}
\item $\size{B_{x_i}\cap V(C_i)}\geq th$, and
\item there are $t$ vertex-disjoint paths in $C_i$ between $B_{x_i}\cap V(C_i)$ and the root clique $R_i$ of $C_i$.
\end{itemize}
Assume without a loss of generality that the node $x_2$ occurs between $x_1$ and $x_3$ on $P$.
We prove the induction statement for $B=B_{x_2}$.

For each $i\in\{1,2,3\}$, we take the $t$ vertex-disjoint paths from $B_{x_i}$ to $R_i$ in $C_i$ and extend them by the matching between $R_i$ and $R$ to obtain $t$ vertex-disjoint paths from $B_{x_i}$ to $R$ in $\blowup{T_h}[R\cup V(C_i)]$.
In particular, there are $t$ vertex-disjoint paths from $B_{x_2}$ to $R$ in $\blowup{T_h}$, as required in the induction statement.
Since $\size{R}=t$, the $t$ paths from $B_{x_1}$ to $R$ and the $t$ paths from $B_{x_3}$ to $R$ together form $t$ vertex-disjoint paths from $B_{x_1}$ to $B_{x_3}$ in $\blowup{T_h}[V(C_1)\cup R\cup V(C_3)]$, which therefore avoid $V(C_2)$.
Since $x_2$ lies between $x_1$ and $x_3$ on $P$, the set $B_{x_2}\setminus V(C_2)$ must contain at least one vertex from each of these $t$ paths.
Thus $\size{B_{x_2}\setminus V(C_2)}\geq t$.
Since $\size{B_{x_2}\cap V(C_2)}\geq th$, we conclude that $\size{B_{x_2}}\geq t(h+1)$, as required in the induction statement.
\end{proof}

\begin{claim}
For any\/ $h\in\setN$, the graph\/ $\blowup{T_h}$ contains no subdivision of a complete binary tree of height\/ $3\max(h+1,\lceil\log_2t\rceil)$.
\end{claim}

\begin{proof}
A simple calculation shows that $T_h$ has $\frac{3^{h+1}-1}{2}$ nodes.
Thus $\size{V(\blowup{T_h})}+1\leq 3^{h+1}t$ and so
\[\log_2\bigl(\size{V(\blowup{T_h})}+1\bigr)\leq\log_2(3^{h+1}t)\leq 3\max(h+1,\lceil\log_2t\rceil).\]
If a graph $G$ contains a subdivision of a complete binary tree of height $c$, then $\size{V(G)}\geq 2^{c+1}-1$ and so $\log_2(\size{V(G)}+1)\geq c+1$.
Therefore, $\blowup{T_h}$ cannot contain a subdivision of a complete binary tree of height $3\max(h+1,\lceil\log_2t\rceil)$.
\end{proof}

\section{Open Problem}

In Theorem~\ref{thm1}, we bound pathwidth by a function of treewidth in the absence of a subdivision of a large complete binary tree.
In~\cite{KR22,CNP21}, the authors bound treedepth by a function of treewidth in the absence of a subdivision of a large complete binary tree and of a long path.
Specifically, Czerwiński, Nadara, and Pilipczuk~\cite{CNP21} proved the following bound.\footnote{This bound is stated in~\cite{CNP21} for the case $t=h=\ell$, but the proof works as well for the general case.}

\begin{theorem}
\label{thm:CNP}
Every graph with treewidth\/ $t$ that contains no subdivision of a complete binary tree of height\/ $h$ and no path of order\/ $2^\ell$ has treedepth\/ $O(th\ell)$.
\end{theorem}

It is natural to ask how large treedepth can be as a function of pathwidth when there is no long path.
We offer the following conjecture.

\begin{conjecture}
\label{conjecture}
Every graph with pathwidth\/ $p$ that contains no path of order\/ $2^\ell$ has treedepth\/ $O(p\ell)$.
\end{conjecture}

This conjecture and Theorem~\ref{thm1} would directly imply Theorem~\ref{thm:CNP}.
We note that an $O(p^2\ell)$ bound on the treedepth follows from Theorem~\ref{thm:CNP}, as every graph with pathwidth $p$ has treewidth at most $p$ and contains no subdivision of a complete binary tree of height $2p+1$.
We also note that an $O(p\ell)$ bound on the treedepth would be best possible; see \cite[Section~7]{CNP21}.

\section*{Acknowledgments}

This research was started at the Structural Graph Theory workshop in Gułtowy, Poland, in June~2019.
We thank the organizers and the other workshop participants for creating a productive working atmosphere.
We also thank the anonymous referees for their helpful comments.
We are particularly grateful to one referee for pointing out a serious error in an earlier version of the paper.

\bibliographystyle{plain}
\bibliography{main}

\end{document}